\documentclass[sigconf,authorversion,nonacm]{acmart}
\AtBeginDocument{%
  \providecommand\BibTeX{{%
    \normalfont B\kern-0.5em{\scshape i\kern-0.25em b}\kern-0.8em\TeX}}}

\usepackage{hyperref}       %
\usepackage{url}            %
\usepackage{booktabs}       %
\usepackage{amsfonts}       %
\usepackage{nicefrac}       %
\usepackage{microtype}      %
\usepackage{lipsum}
\usepackage{dsfont}
\usepackage{graphicx}
\usepackage{float}
\usepackage{amsmath}
\usepackage{amsthm}
\usepackage{enumitem}
\usepackage{algorithm}
\usepackage{xspace}
\usepackage{multirow}
\usepackage[export]{adjustbox}
\usepackage[framemethod=TikZ]{mdframed}
\usepackage[noend]{algpseudocode}
\usepackage[strings]{underscore}

\newtheorem{theorem}{Theorem}[section]
\newtheorem{lemma}[theorem]{Lemma}
\newtheorem{problem}[theorem]{Problem}

\newtheorem{definition}[theorem]{Definition}
\newtheorem{observation}[theorem]{Observation}
\newtheorem{corollary}[theorem]{Corollary}

\newtheorem{fact}[theorem]{Fact}

\renewenvironment{quote}
  {\list{}{\rightmargin=0.8cm \leftmargin=0.8cm}%
   \item\relax}
  {\endlist}

\newcommand{\uniffrom}{\ensuremath{\overset{r}{\leftarrow}}}

\RenewDocumentCommand\P{ m g }{
  \ensuremath{
    \IfNoValueTF{#2}
     {\Pr [ #1 ]}
    {\Pr_{#1}[#2]}
  }
}
\DeclareMathOperator*{\Expectation}{\mathbb{E}}
\NewDocumentCommand\E{ m g }{
  \ensuremath{
    \IfNoValueTF{#2}
    {\Expectation \left[#1\right]}
    {\Expectation_{#1}\left[#2\right]}
  }
}
\DeclareMathOperator*{\Variance}{\mathbb{V}\mathrm{ar}}
\NewDocumentCommand\Var{ m g }{
  \ensuremath{
    \IfNoValueTF{#2}
    {\Variance \left[#1\right]}
    {\Variance_{#1}\left[#2\right]}
  }
}

\newcommand{\UsrData}[1]{{\text{st}}_{#1}}
\newcommand{\UsrDerivative}[1]{X_{#1}}
\newcommand{\PorpulationSum}[1]{{a}[#1]}
\newcommand{\EstPorpulationSum}[1]{\hat{{a}}[#1]}

\NewDocumentCommand\dyadicInterval{ g g }{
  \ensuremath{
    \IfNoValueTF{#2}
    { \IfNoValueTF{#1}{ \mathcal{I} }{ \mathcal{I}_{#1, * } } }
    { \mathcal{I}_{#1, #2} }
  }
}
\NewDocumentCommand\dyadicIntervalSet{g}{
  \ensuremath{
    \IfNoValueTF{#1}{ \textit{ISet} }{ \textit{ISet}[{#1}] }
  }
}
\newcommand{\dyadicDecomposition}[1]{\mathcal{C}(#1)}
\NewDocumentCommand\partialsum{gg}{
  \ensuremath{
    \IfNoValueTF{#2}
    { S ({#1} ) }
    { S_{#1} ({#2}) }
  }
}
\newcommand{\EstPartialsum}[1]{ \hat{S}({#1}) }
\newcommand{\IntSet}[2]{[{#1}\,.\,.\,{#2}]}
\newcommand{\ldp}{{\bf LDP}\xspace}
\newcommand{\targetOutputs}{\mathcal{G}}
\newcommand{\indicator}[1]{\mathds{1}_{\left[#1\right]}}

\newcommand{\norm}[1]{\left\Vert {#1} \right\Vert}
\newcommand{\PAREN}[1]{{\left( {#1} \right)}}
\newcommand{\paren}[1]{{( {#1} )}}
\newcommand{\bracket}[1]{{[ {#1} ]}}

\newcommand{\card}[1]{\left| {#1} \right|}
\newcommand{\set}[1]{\left\{ {#1} \right\}}
\newcommand{\eps}{\varepsilon}
\newcommand{\tildeeps}{\tilde{\varepsilon}}

\newcommand{\supp}[1]{\text{supp}(#1)}
\newcommand{\Annulus}[1]{\textbf{Ann}(#1)}

\newcommand{\cGap}{c_{\text{gap}}}
\newcommand{\pMax}{p_{\text{max}}}
\newcommand{\pMin}{p_{\text{min}}}

\newcommand{\NewLb}{\text{LB}}
\newcommand{\NewUb}{\text{UB}}
\newcommand{\nnz}{\mathsf{nnz}}
\newcommand{\PrOut}{P_\text{out}}
\newcommand{\numerator}[1]{\mathcal{N}_{#1}}
\newcommand{\denominator}[1]{\mathfrak{D}_{#1}}
\newcommand{\init}{\textit{init}}
\newcommand{\rndmzrPropI}{{\bf Property I}\xspace}
\newcommand{\rndmzrPropII}{{\bf Property II}\xspace}
\newcommand{\rndmzrPropIII}{{\bf Property III}\xspace}

\newcommand{\ourRandomizer}{\textit{FutureRand}\xspace}
\def\algoServer{\mathcal{A}_{ \text{svr} }}

\def\algoClient{\mathcal{A}_{ \text{clt} }}

\def\pavg{p_{\text{avg} } }

\newcommand{\cA}{\mathcal{A}}

\newcommand{\cD}{\mathcal{D}}

\newcommand{\cM}{\mathcal{M}}

\newcommand{\cR}{\mathcal{R}}

\newcommand{\cU}{\mathcal{U}}

\newcommand{\cY}{\mathcal{Y}}

\newcommand{\pmbomega}{\pmb{\omega}}

\newcommand{\N}{\mathbb{N}}
\newcommand{\R}{\mathbb{R}}

\newcommand{\mbbZ}{\mathbb{Z}}

\newif\ifcomment
\commenttrue

\ifcomment
\newcommand{\tony}[1]{\textcolor{brown}{[TONY: #1]}}
\newcommand{\hao}[1]{\textcolor{blue}{[HAO: #1]}}
\newcommand{\olya}[1]{\textcolor{purple}{[Olya: #1]}}

\definecolor{DarkGreen}{rgb}{0.1,0.5,0.1}
\newcommand{\suggest}[1]{\textcolor{DarkGreen}{#1}}
\definecolor{DarkGreen}{rgb}{0.1,0.5,0.1}

\else
\newcommand{\tony}[1]{}
\newcommand{\hao}[1]{}
\newcommand{\olya}[1]{}
\fi

\begin{document}
\fancyhead{}

\title{Randomize the Future: Asymptotically Optimal Locally Private Frequency Estimation Protocol for Longitudinal Data}

\author{Olga Ohrimenko}
\email{oohrimenko@unimelb.edu.au}
\affiliation{%
  \institution{The University of Melbourne$^\dagger$}
  \authornote{
    School of Computing and Information Systems
    \vspace{-1.5mm}
  }
  \country{Australia}
}

\author{Anthony Wirth}
\email{awirth@unimelb.edu.au}
\affiliation{%
  \institution{The University of Melbourne$^\dagger$}
  \country{Australia}
}

\author{Hao Wu}
\email{whw4@student.unimelb.edu.au}
\affiliation{%
  \institution{The University of Melbourne$^\dagger$}
  \country{Australia}
}

\renewcommand{\shortauthors}{}

\begin{abstract}

    Longitudinal data tracking under Local Differential Privacy (\ldp) is a challenging task.
    Baseline solutions that repeatedly invoke a protocol designed for one-time computation lead to linear decay in the privacy or utility guarantee with respect to the number of computations.
    To avoid this, the recent approach of \citeauthor{EFMRTT19}~(\citeyear{erlingsson2020amplification}) exploits the potential sparsity of user data that  changes only infrequently.
    Their protocol targets the fundamental problem of frequency estimation for longitudinal binary data, 
    with~$\ell_\infty$ error of~$O ( (1 / \eps) \cdot (\log d)^{3 / 2} \cdot k \cdot \sqrt{ n  \cdot \log ( d /  \beta ) } )$, where~$\eps$ is the privacy budget,~$d$ is the number of time periods,~$k$ is the maximum number of changes of user data, and~$\beta$ is the failure probability.
    Notably, the error bound scales polylogarithmically with~$d$, but linearly with~$k$.
    
    In this paper, we break through the linear dependence on~$k$ in the estimation error.
    Our new protocol has error~$O ( (1 / \eps) \cdot (\log d) \cdot \sqrt{  k \cdot n  \cdot \log ( d /  \beta ) }  )$, matching the lower bound 
    up to a logarithmic factor.
    The protocol is an online one, that outputs an estimate at each time period. 
    The key breakthrough is a new randomizer for sequential data, \ourRandomizer, with two key features. 
    The first is a composition strategy that correlates the noise across the non-zero elements of the sequence.
    The second is a pre-computation technique which, by exploiting the symmetry of input space, enables the randomizer to output the results on the fly, without knowing future inputs.
    Our protocol closes the error gap between existing online and offline algorithms.

\end{abstract}


\keywords{Local Differential Privacy, Longitudinal Data, Frequency Estimation \vspace{-4mm}}

\maketitle

\section{Introduction}

Frequency estimation underlines a wide range of applications in data mining and machine learning (for example, learning users’ preferences, uncovering commonly used phrases, and finding popular URLs).
However, the data collected for the frequency analysis can contain sensitive personal information such as income, gender, health information, etc. 
To protect this information from the data collector, the local model of differential privacy (\ldp) has been deployed by companies including Google~\citep{EPK14, FPE16}, Apple~\citep{D.P.Apple17} and Microsoft~\citep{DKY17}.
In \ldp, each user perturbs their data locally before reporting it to the (untrusted) data collector (aka server) for analysis.

Existing solutions typically focus on one-time computation.
However practical applications often involve continuous monitoring in order to discover trends over time. 
For example, search-engine providers keep track of popular URLs. The naive solution of repeated computation leads to a rapid degradation of privacy guarantee, that scales linearly with the number of computations~\citep{DR14}.
But such degradation is unnecessary, when  users' data changes infrequently.
For example,  a list of frequently visited URLs by a user changes little everyday.

The observation of infrequent data changes is captured by~\citeauthor{EFMRTT19}~\cite{erlingsson2020amplification}.
Noting that many \ldp algorithms~\citep{BNS19, BNST20, EPK14, NXYHSS16} require each client to send just one bit to the server, they %
formulated the following longitudinal data tracking problem for Boolean data.

\begin{quote}
    \textbf{Research Problem:}
    Given a set of~$n$ users, each holding a Boolean value that changes at most~$k$ times across~$d$ time periods; the server needs to report the number of users holding Boolean value~$1$ at each time period.
\end{quote}
The problem is presented in an online setting; in an offline setting, the server reports an estimate only after~$d$ time periods.
Though here we study the problem over Boolean data, our algorithm
can be adapted to solve frequency estimation and heavy hitter problems in richer domains via existing techniques~\citep{BNS19, BNST20, JRUW18, EPK14, FPE16}.

For privacy budget~$\eps$ and~failure probability~$\beta$, \citeauthor{EFMRTT19} described a protocol that achieves maximum estimation error~$O ( (1 / \eps) \cdot (\log d)^{3 / 2} \cdot k \cdot \sqrt{ n  \cdot \log ( d /  \beta ) } )$~\cite{erlingsson2020amplification}, and scales only linearly with~$k$ and not~$d$.
In this paper we study the following question:
\begin{quote}
    \textbf{Research Question:}
    Can the estimation error for the (online) longitudinal data tracking problem above be reduced to sub-linear in~$k$?
\end{quote}
Besides improvement on previous work in terms of the error, answering this question can
close the gap between error guarantees of the online algorithm and the lower bound
of~$\Omega ( (1 / \eps) \cdot \sqrt{  k \cdot n  \cdot \log ( d /  k ) } )$
that was recently presented in~\citep{zhou2021locally}.

\vspace{1mm}
\noindent {\bf Our Contributions.}
In this work, we provide a new \ldp protocol whose error scales with the square root of the number of data changes.
Specifically, it achieves error~
$$
    O \big( (1 / \eps) \cdot (\log d) \cdot \sqrt{  k \cdot n  \cdot \log ( d /  \beta ) }  \big).
$$ 
Our protocol builds on a standard technique for converting the original data sequence into a sparse one~\citep{erlingsson2020amplification}, hierarchical aggregation scheme for releasing continual data~\citep{DworkNPR10}, and a new randomizer, \ourRandomizer, for sequential data. 
The \ourRandomizer has two key components, a composed randomizer~$\tilde{\cR} : \set{-1, 1}^k \rightarrow \set{-1, 1}^k$ for the non-zero coordinates of the input sequence, and a pre-computation technique that enables~$\tilde{\cR}$ to handle online inputs. 
The main properties of~$\tilde{\cR}$, which play vital role in establishing privacy guarantee and achieving a~$\sqrt{k}$ estimation error, are stated below: assuming that~$\eps \le 1$
\begin{itemize}
    \item 
        There exit~$\pMin', \pMax' \in \paren{0, 1}$ with~$\pMax' \le e^\eps \cdot \pMin'$, such that for each input sequence~$b \in \set{-1, 1}^k$, and each sequence~$s \in \set{-1, 1}^k$, then
        \begin{equation*}
            \P{ \tilde \cR(b) = s } \in \bracket{ \pMin', \pMax' }. 
        \end{equation*}
    
    \item 
        For a given input~$b$, denote~$\tilde{b} \doteq \tilde{\cR}(b)$ the output of~$\tilde{\cR}$.
        There exists some~$\cGap \in \Omega \paren{ \eps / \sqrt{k} }$, such that for each input~$b \in \set{-1, 1}^k$, and for each~$i \in \bracket{k}$,
        \begin{equation*}
            \P{ \tilde{b}_i = b_i } - \P{ \tilde{b}_i = -b_i } = \cGap \in \Omega \paren{ \eps / \sqrt{k} }. 
        \end{equation*}
\end{itemize}

$\tilde{\cR}$ builds on the composed randomizer proposed by~\citeauthor{BNS19}~\citep{BNS19}.
The design in~\citep{BNS19} focused on preserving the statistical distance between the distribution of the output of the composed randomizer, and joint distribution of~$k$ independent randomized responses.
This difference in the problem setting requires non-trivial changes in terms of parameters, assumptions and analysis.
The proof in~\citep{BNS19} relies extensively on the concentration and anti-concentration inequalities;
in comparison, our proof avoids this and investigates the inherent structure of the problem.
Finally, \citeauthor{BNS19}'s original design~\citep{BNS19} applies only to offline inputs. 
Our paper develops the pre-computation technique for converting it into an online one.

\vspace{1mm}
\noindent {\bf Organization.}
Our paper is organized as follows.
Section~\ref{sec: problem definition} introduces the problem formally. 
Section~\ref{sec: prelimary} discusses the preliminaries for our protocol.
Section~\ref{sec: framework} develops an algorithmic framework for longitudinal data tracking.
Section~\ref{sec: our randomizer} introduces the \ourRandomizer.
Section~\ref{sec: review} summarizes the related works.

\section{Problem Definition }
\label{sec: problem definition}

We present formally the longitudinal data tracking problem first introduced by~\citeauthor{EFMRTT19}~\cite{erlingsson2020amplification}.
There is a server and a set of~$n$ users, each holding a Boolean data item that changes at most~$k$ times across the~$d$ time periods. 
Without loss of generality, we assume that~$d$ is a power of~$2$.
For user $u \in [n]$, denote the value sequence of its Boolean data across the~$d$ time periods as
$$
    \UsrData{u} = ( \UsrData{u}[1], \ldots, \UsrData{u} [d]  ) \in \set{ 0, 1 }^d\,.
$$
For each time~$t \in [d]$, denote the number of users with value~$1$ by
\begin{equation}
    \label{equa: def porpulation sum}
    \PorpulationSum{t} \doteq \sum_{ u \in [n] } \UsrData{u} [t].
\end{equation}

\begin{definition}[$(\alpha, \beta)$-Accurate Protocol]
    Given~$\alpha > 0$, and~$\beta \in (0, 1)$, a server-side algorithm is called an~$(\alpha, \beta)$-accurate protocol for longitudinal data collection if it outputs at each time~$t \in [d]$, an estimate~$\EstPorpulationSum{t}$ of~$\PorpulationSum{t}$, such that
    $$
        \begin{array}{c}
            \P{ \max_{t \in [d] } \left| \EstPorpulationSum{t} - \PorpulationSum{t} \right| > \alpha }
            \le \beta\,.
        \end{array}
    $$
\end{definition}

\begin{definition}[Differential Privacy \citep{DR14}] \label{def: Differential Privacy}
    Let $\cA: \cD \rightarrow \cY$ be a randomized algorithm.
    Algorithm~$\cA$ is called  $\eps$-differentially private if for all $v, v' \in \cD$ and all (measurable) $Y \subseteq \cY$,
    \begin{equation} \label{ineq: def private algo}
        \P{ \cA (v) \in Y } \le e^\eps \cdot \Pr [ \cA (v') \in Y ]\,.
    \end{equation}
    The parameter~$\eps$ characterizes the similarity of the output distributions of~$\cA (v)$ and~$\cA (v')$, and is called the \emph{privacy budget}.
\end{definition}

A data-collecting protocol is called~$\eps{}$-locally differentially private (\ldp) if each user reports only a version of their data perturbed locally by an $\eps{}$-differentially private algorithm.

\begin{problem}[Private Longitudinal Data Collection] Given privacy parameter~$\eps{}$, and failure probability $\beta$, design an $\eps$-locally differentially private~$(\alpha, \beta)$-accurate protocol for longitudinal data collection, that minimizes~$\alpha$.
\end{problem}

Our goal is to achieve an~$\alpha \in O ( (1 / \eps) \cdot (\log d) \cdot \sqrt{  k \cdot n  \cdot \log ( d /  \beta ) }  )$.

\section{Preliminaries}
\label{sec: prelimary}

    We use the following notation.
    For every~$i, j \in \N$, such that $i \le j$, ~$\IntSet{i}{j}$ denotes the set of integers~$\set{i, \ldots, j}$.
    For each~$j \in \N^+$, set~$\bracket{ j }$ refers to $\IntSet{1}{j}$.

    As the data of each user~$u$ changes at most~$k$ times, we introduce the following transformation of~$\UsrDerivative{u}$ into a sparse vector, which has at most~$k$ non-zero coordinates.

\begin{definition}[Data Derivative~\citep{erlingsson2020amplification}]
    For each user~$u$ and for all~$t \in \bracket{ d }$, let $\UsrDerivative{u}\bracket{ t } \doteq \UsrData{u}\bracket{ t } - \UsrData{u}\bracket{ t - 1 }$, where for convenience, we assume that $\UsrData{u}\bracket{ 0 } = 0$, so that~$\UsrDerivative{u}\bracket{ 1 }$ is well defined.
    The discrete derivative of user data is
    $$
        \UsrDerivative{u} = \paren{ \UsrDerivative{u}\bracket{ 1 }, \ldots, \UsrDerivative{u}\bracket{ d } } \in \set{ -1, 0, 1 }^d.
    $$
\end{definition}

For example, if~$\UsrData{u} = \paren{ 0, 1, 1, 0}$, then~$\UsrDerivative{u} = \paren{ 0, 1, 0, -1}$.
Observe that for each~$t \in \bracket{ d }$,
$
    \UsrData{u}\bracket{ t } = \sum_{t' \in \bracket{ t } } \UsrDerivative{u}\bracket{  t'  }.
$
We are interested in the cumulative data changes over intervals whose length is a power of~$2$: such intervals are called \emph{dyadic}.

\begin{definition}[Dyadic Intervals]
    For each~$h \in \IntSet{0}{\log d}$, and each~$j \in \bracket{ d / 2^h }$, define the dyadic interval~$\dyadicInterval{h}{j} \doteq \{ (j - 1) \cdot 2^h + 1, \ldots, j \cdot 2^h \}$, where~$h$ is called its \emph{order}.
    Denote the collection of dyadic intervals of order~$h$ by
    $$
        \dyadicIntervalSet{h} \doteq \set{ \dyadicInterval{h}{j} : j \in \bracket{ d / 2^h  } }\,,
    $$
    and the collection of all dyadic intervals by
    $$
        \dyadicIntervalSet \doteq \cup_{ h \in \IntSet{0}{\log d} } \dyadicIntervalSet{h}\,.
    $$ 
    Interval~$\dyadicInterval$ represents an element from~$\dyadicIntervalSet$.
\end{definition}

\begin{figure}[t]
	\includegraphics[width=0.95\linewidth,right]{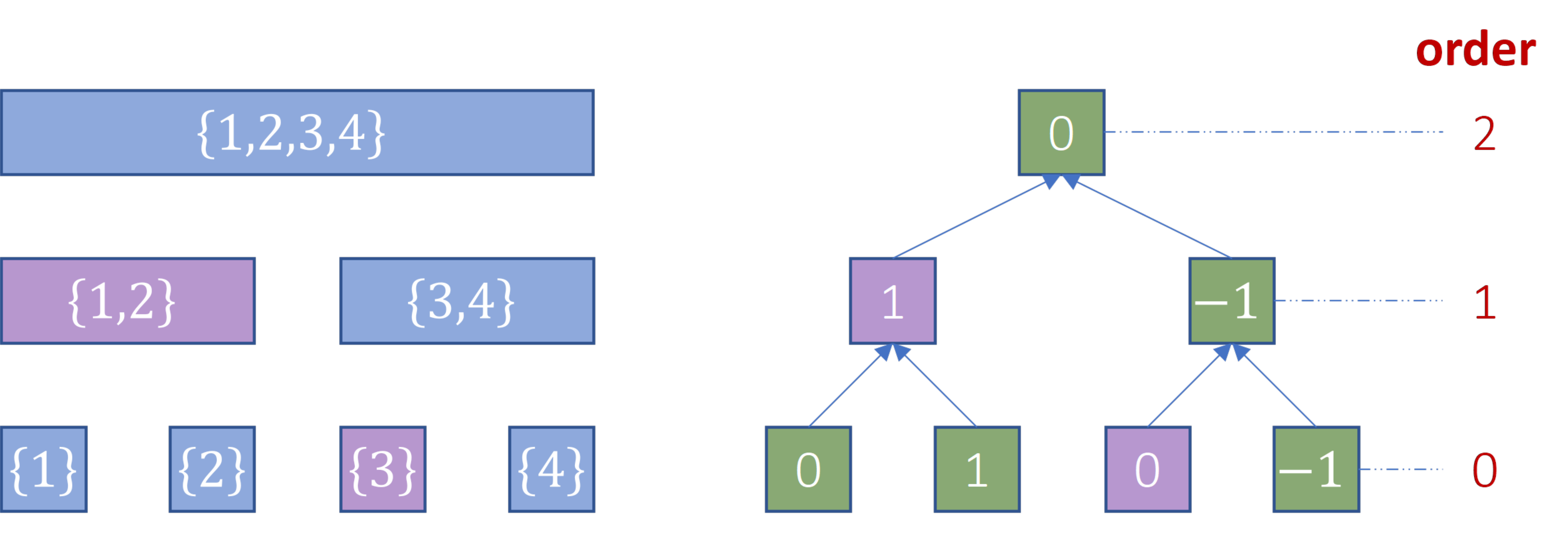}
	\vspace{-4mm}
	\caption{ 
	An example with~$d = 4$ and~$k = 2$.
	The left-hand side enumerates all dyadic intervals defined on~$\bracket{d}$, with the nodes highlighted in purple representing the dyadic decomposition~$\dyadicDecomposition{3}$ of the interval~$\bracket{3}$.
	The right-hand side enumerates all partial sums for the discrete derivative~$\UsrDerivative{u} = \paren{0, 1, 0, -1}$, 
	with the nodes highlighted in purple representing the partial sums associated with dyadic intervals in~$\dyadicDecomposition{3}$. \copyright Olga Ohrimenko, Anthony Wirth, Hao Wu
	}
	\label{fig:dyadic}
	\vspace{-4mm}
\end{figure}

\begin{example}
    All possible dyadic intervals defined on 
    interval~$\bracket{4}$ include
    $\dyadicInterval{0}{1} = \set{1}$,
    $\dyadicInterval{0}{2} = \set{2}$,
    $\dyadicInterval{0}{3} = \set{3}$,
    $\dyadicInterval{0}{4} = \set{4}$,
    $\dyadicInterval{1}{1} = \set{1, 2}$,
    $\dyadicInterval{1}{2} = \set{3, 4}$,
    $\dyadicInterval{2}{1} = \set{1, 2, 3, 4}$. 
\end{example}

\begin{definition}[Partial Sum]
    For each~$h \in \IntSet{0}{\log d}$, and each~$j \in \bracket{ d / 2^h }$, 
    define  for each user
    the partial sum associated with~$\dyadicInterval{h}{j}$ to be
    \vspace{-2mm}
    \begin{equation} 
        \label{equa: def of partial sum}
        \partialsum{u}{\dyadicInterval{h}{j}} \doteq \sum_{ t \in \dyadicInterval{h}{j} } \UsrDerivative{u}\bracket{ t }, \, \forall u \in [n].
    \end{equation}
    \vspace{-1mm}
    and their sum as
    \vspace{-1mm}
    \begin{equation}
        \label{equa: def sum of partial sum}
        \partialsum{\dyadicInterval{h}{j} } \doteq \sum_{ u \in \bracket{ n } } \partialsum{u}{\dyadicInterval{h}{j} }
    \end{equation}
    \vspace{-1mm}
    Here, $h$ refers to the \emph{order} of a partial sum. 
\end{definition}

\begin{example}
    Suppose that~$\UsrDerivative{u} = \paren{0, 1, 0, -1}$. 
    The possible partial sums are
    $\partialsum{u}{ \dyadicInterval{0}{1}} = 0$,
    $\partialsum{u}{ \dyadicInterval{0}{2}} = 1$,
    $\partialsum{u}{ \dyadicInterval{0}{3}} = 0$,
    $\partialsum{u}{ \dyadicInterval{0}{4}} = -1$,
    $\partialsum{u}{ \dyadicInterval{1}{1}} = 1$,
    $\partialsum{u}{ \dyadicInterval{1}{2}} = -1$,
    $\partialsum{u}{ \dyadicInterval{2}{1}} = 0$. 
\end{example}

Note that if~$\partialsum{u}{\dyadicInterval{h}{j}} \neq 0$, then there is at least one~$t \in \dyadicInterval{h}{j}$ for which~$\UsrDerivative{u}\bracket{t} \neq 0$.
Since dyadic intervals of order~$h$ are disjoint, and there are at most~$k$ non-zero~$\UsrDerivative{u}\bracket{t}$ over all~$t \in [d]$, we have this observation:

\begin{observation} \label{observation: upper bound for non zero elements each order}
     For each ~$h \in \IntSet{0}{\log d}$, there can be at most~$k$ indices~$j \in \bracket{ d / 2^h }$ for which~$\partialsum{u}{\dyadicInterval{h}{j}} \neq 0$.
\end{observation}

Combining Equation~(\ref{equa: def of partial sum}) and the fact that~$\UsrDerivative{u}\bracket{ t } = \UsrData{u}\bracket{ t } - \UsrData{u}\bracket{ t - 1 }$, and that~$\dyadicInterval{h}{j} = \{ (j - 1) \cdot 2^h + 1, \ldots, j \cdot 2^h \}$, we observe:

\begin{observation} \label{observation: basic partial sum properties}
    For each~$u \in [n]$, and~$\dyadicInterval{h}{j} \in \dyadicIntervalSet$, we have
    \begin{equation} \label{equa: range of partial sum}
    \partialsum{u}{\dyadicInterval{h}{j}} 
        = \UsrData{u}\bracket{ j \cdot 2^h } - \UsrData{u}\bracket{ (j - 1) \cdot 2^h  } 
        \in \set{-1, 0, 1}\,.
    \end{equation}
\end{observation}

Recall that~$\UsrData{u}\bracket{ t } = \sum_{t' \in \bracket{ t } } \UsrDerivative{u}\bracket{  t'  }$.
It can be viewed as the cumulative data change from time~$1$ to time~$t$.
If we  decompose the interval~$\bracket{ t }$ into a sequence of disjoint dyadic intervals, then~$\UsrData{u}\bracket{ t }$ can be expressed as the sum of cumulative changes over these intervals. 

\begin{fact}(Dyadic Decomposition) \label{fact: nodes be in different level}
    For each~$t \in \bracket{ d }$, the interval~$\bracket{ t }$ can be decomposed into a minimum collection, $\dyadicDecomposition{t}$, of at most~$\lceil \log t \rceil$ disjoint dyadic intervals with distinct orders.
\end{fact}

For example, the dyadic decomposition of the interval~$\bracket{3}$ is given by~$\set{ \set{1,2}, \set{3} }$.
See Figure~\ref{fig:dyadic} for the illustration.
In general, for~$1 \le \ell \le r \le d$, the interval~$\IntSet{\ell}{r}$ can also be decomposed into a minimum collection of at most~$\lceil 2 \cdot  \log (r - \ell + 1)
\rceil$
disjoint dyadic intervals.
But these intervals might not have distinct orders.
For example, the decomposition of the interval~$\IntSet{2}{3}$ is given by~$\set{\set{2}, \set{3}}$.
Fact~\ref{fact: nodes be in different level} holds since the interval~$\bracket{t}$ starts from~$1$.

Via the definition of~$\dyadicDecomposition{t}$, we have~$\UsrData{u}\bracket{ t } = \sum_{ \dyadicInterval{h}{j} \in \dyadicDecomposition{t} } \partialsum{u}{\dyadicInterval{h}{j}}$.
Summing over~$u \in [n]$, and combining this with Equation~(\ref{equa: def porpulation sum}),~(\ref{equa: def of partial sum}) and~(\ref{equa: def sum of partial sum}), we have

\begin{observation}
    For each time~$t \in [d]$, 
    \begin{equation}
        \label{equa: decomposition of porpulatioin sum}
        \PorpulationSum{t} 
        = \sum_{u \in \bracket{n} } \sum_{ \dyadicInterval{h}{j} \in \dyadicDecomposition{t} }  \partialsum{u}{\dyadicInterval{h}{j}}
        = \sum_{ \dyadicInterval{h}{j} \in \dyadicDecomposition{t} } \partialsum{\dyadicInterval{h}{j}}.
    \end{equation}    
\end{observation}

\section{Framework} \label{sec: framework}

In this section, we introduce an \ldp protocol 
that achieves our proposed error guarantee.
It improves the previously known bound~\citep{erlingsson2020amplification} by a factor of~$\sqrt{k}$. %
The main result of the paper is summarized as follows:

\begin{theorem} \label{theorem: property of further improved private algorithm}
    Assuming that $\eps \le 1$, and $ \eps^{-1} \cdot \paren{ \log d } \cdot \sqrt{ k \cdot \ln \paren{ d / \beta } }$ $\le \sqrt{n}$, there is an $\eps$-local differentially private protocol, such that, with probability at least $1 - \beta$, its~$d$ estimates~$\{\EstPorpulationSum{t}\}$ of~$\{\PorpulationSum{t}\}$ satisfy
    $$
        \max_{t \in [d] } \card{ \EstPorpulationSum{t} - \PorpulationSum{t} }
        \in
        O \left( \frac{ \log d } {\eps} \cdot \sqrt{  k \cdot n  \cdot \ln \frac{d}{\beta} }  \right)\,.
    $$
\end{theorem}

The assumption $\eps^{-1} \cdot \paren{ \log d } \cdot \sqrt{ k \cdot \ln \paren{ d / \beta } } \le \sqrt{n}$ ensures that ~$\eps^{-1} \cdot \paren{ \log d } \cdot \sqrt{ k \cdot n \cdot \ln \paren{ d / \beta } } \le n$, avoiding a trivial bound.

\vspace{3mm}
\subsection{Overview}

Our framework is split between the user (or client) and the server.
Each user reports perturbed data to the server, who aggregates
the data from all users and computes the required statistics while adjusting for the noise in the reports.

As a warm up consider the following (non-private) naive protocol.
Here, each user~$u \in \bracket{n}$ computes and reports (to the server) each partial sum~$\partialsum{u}{ \dyadicInterval{h}{j} } $ immediately after it has all the data needed for the computation.
In particular, since the last number in the dyadic interval~$\dyadicInterval{h}{j}$ is~$j \cdot 2^{h_u}$, the last data point needed to compute~$\partialsum{u}{ \dyadicInterval{h}{j} }$ is~$\UsrDerivative{u} \bracket{ j \cdot 2^{h_u} }$.
Based on the reports of the partial sums, according to Equation~(\ref{equa: decomposition of porpulatioin sum}), the server can obtain for each~$t \in [d]$ a precise value of~$\PorpulationSum{t}$.

The naive protocol does not provide LDP guarantees. To this end, our algorithms build on a combination of the following two reporting techniques.

\vspace{2mm}
\noindent {\bf Reporting based on Sampling.}
Suppose, instead of reporting every 
partial sum, each user~$u \in \bracket{n}$ samples an integer~$h_u \in \IntSet{0}{\log d}$ uniformly at random, and reports only the partial sums with order~$h_u$.
The server no longer has exact values of the~$\PorpulationSum{t}, \forall t \in \bracket{d}$.
However, the server can easily construct an unbiased estimate.
For each~$h \in \IntSet{0}{\log d}$, and each~$j \in \bracket{ d / 2^h }$, let~$z_u[h, j]$ be the server's estimate of~$\partialsum{u}{ \dyadicInterval{h}{j} }$. 
If it sets~$z_u[h, j] = (1 + \log d) \cdot \indicator{h_u = h} \cdot \partialsum{u}{ \dyadicInterval{h}{j} }$, where~$\indicator{h_u = h}$ is the indicator for the event~$h_u = h$,
then 
~$z_u[h, j ]$ is an unbiased estimate of~$\partialsum{u}{ \dyadicInterval{h}{j} }$.
Via linearity of expectation, replacing~$\partialsum{u}{ \dyadicInterval{h}{j} }$ in Equation~(\ref{equa: decomposition of porpulatioin sum}) with~$z_u[h, j ]$ gives an unbiased estimator~$\PorpulationSum{t}$.

\vspace{2mm}
\noindent {\bf Reporting with Perturbation.} 
In the previous paragraph, conditioned on~$h_u = h$, user~$u \in \bracket{n}$ reports a sequence of $L \doteq d / 2^h$ partial sums with order~$h$ to the server, each taking value in~$\set{-1, 0, 1}$ according to Equation~(\ref{equa: range of partial sum}).
To achieve \ldp with this approach, the user~$u$ should invoke some randomizer~$\cM$, to perturb their partial sums before reporting.
This also requires a change in computing the estimators~$z_u[h, j]$ with respect to~$\cM$.
Our randomizer,~$\cM$, exploits sparsity: there are at most~$k$ non-zero partial sums with order~$h$
(Observation~\ref{observation: upper bound for non zero elements each order}).

\subsection{Client Side} \label{subsec: client side}

{\bf Randomizer.}
Let~$\cM$ be a client-side randomizer. Our protocol relies on a randomizer 
$\cM$ that provides the following functionalities and has three properties described below.

$\cM$ has an (optional) initialization phase~$\cM.\init(L, k, \eps)$, with parameters the length of the input~$L$, the maximum number of non-zero elements in the input~$k$, and the privacy budget~$\eps$.
In this phase, $\cM$ can perform some pre-computation whose result is kept for reference later.

During the protocol, $\cM$ takes as input a sequence $v_1, \ldots, v_L$ where each value is in $\{-1, 0, 1\}$, corresponding to user's data, and outputs a sequence $\cM^{(1)} (v_1), \ldots, \cM^{(L)} (v_L) \in \{-1, 1\}$.
The output of~$\cM^{ (j) }(v_j)$ may depend not only on the input~$v_j$ and the randomness of~$\cM^{ (j) }$, but also on the pre-computation result, the past inputs~$v_1, \dots, v_{j - 1}$ and outputs~$\cM^{(1)} (v_1), \ldots, \cM^{(j - 1)} (v_{j - 1})$.

Finally,~$\cM$ should satisfy the following three properties.

\vspace{2mm}
\begin{mdframed}[roundcorner=3.5pt]
    \begin{description}[leftmargin = *]
        \item [\rndmzrPropI.] 
            Given~$L$,~$k$ and~$\eps \le 1$, there exists~$\pMin, \pMax \in \paren{0, 1}$, such that~$\pMax / \pMin \le e^{\eps}$, and that for each $k$-sparse input sequence~$v_1, \ldots, v_L \in \set{-1, 0, 1}$, and each sequence~$w_1, \ldots,$ $w_L \in \set{-1, 1}$,
            \begin{equation} \label{ineq: randomizer output prob range}
                \hspace{-1mm}
                \P{ M^{ (1) } (v_1) = w_1, \ldots, M^{ (L) } (v_L) = w_L }
                \in
                \bracket{ \pMin, \pMax}\,.
            \end{equation}
        \item [\rndmzrPropII. ]
            There exits~$\cGap \in \paren{0, 1}$, such that for all~$j \in \bracket{L}$, if $v_j \neq 0$,
            \begin{equation} \label{equa: uniform gap}
                \P{ \cM^{ (j) }(v_j) = v_j } - \P{ \cM^{ (j) }(v_j) = - v_j }
                  = \cGap\,.          
            \end{equation}
        \item [\rndmzrPropIII.]
        For all~$j \in \bracket{L}$, if $v_j = 0$,
            \begin{equation} \label{equa: equal output}
                \P{ \cM^{ (j) }(v_j) = 1 } = \P{ \cM^{ (j) }(v_j) = - 1 } = \frac{1}{2}\,.          
            \end{equation}
    \end{description}
\end{mdframed}

\begin{example} \label{example: naive randomizer}
The following randomizer~\citep{W65} satisfies these three properties.
    Suppose that~$\cM$ perturbs each coordinate independently such that: if $v_j \neq 0$, then~$\P{ \cM^{ (j) }(v_j) = v_j } = e^{\eps / k} / \paren{ e^{\eps / k} + 1}$ and~$ \P{ \cM^{ (j) }(v_j) = - v_j } = 1 / \paren{ e^{\eps / k} + 1}$; if $v_j = 0$,  then~$\cM^{ (j) }(v_j)$ outputs~$-1$ or~$1$, uniformly at random. 
    Now, \rndmzrPropII and \rndmzrPropIII are clearly satisfied with~$\cGap = \paren{ e^{\eps / k} - 1} / \paren{ e^{\eps / k} + 1}$.
    Finally, it can be verified that \rndmzrPropI is satisfied, with~$\pMin = 2^{ - (L - k) } \cdot \paren{ e^{\eps / k} + 1}^{-k}$ and~$\pMax = e^{\eps} \cdot 2^{ - (L - k) } \cdot \paren{ e^{\eps / k} + 1}^{-k}$.
    \end{example}
    This simple randomizer does not perform pre-computation in the initialization phase, and for each~$j \in \bracket{L}$, the result of~$\cM^{ (j) }(v_j)$ does not depend on historical inputs or outputs.

Intuitively, \rndmzrPropI states that, regardless of the input sequence, each sequence in~$\set{-1, 1}^L$ is output by~$\cM$ with similar probability (up to a factor of $e^\eps$), by which the client-side algorithm is differentially private.
\rndmzrPropII ensures that $\cM$ preserves each non-zero coordinate of the input sequences with a common probability; indeed,~$\cGap$ is the common difference between the coordinate preservation and reversal probabilities.
\rndmzrPropIII requires that~$\cM$ outputs~$1$ and~$-1$ with equal probability for each zero coordinate. 
Based on \rndmzrPropII and \rndmzrPropIII: 
\begin{observation} For~$v_j \in \set{-1, 0, 1}$, it always holds that 
    \begin{equation} \label{equa: unbias of each coordinate of M}
        \E{ \paren{ \cGap }^{-1} \cdot \cM^{ (j) }(v_j) }
        = v_j.
    \end{equation}
\end{observation}

Since~$\cM^{ (j) }(v_j)$ is in~$[-1, 1]$; multiplying it by~$\paren{ \cGap }^{-1}$ amplifies the range to~$\bracket{- \paren{ \cGap }^{-1}, \paren{ \cGap }^{-1} }$. 
As will be explained in Section~\ref{subsec: server side}, this range plays a vital role in utility analysis. 
The smaller this range, the better utility we can obtain. Since~$\eps \le 1$, the randomizer in Example~\ref{example: naive randomizer} guarantees~$\paren{ \cGap }^{-1} = \paren{ e^{\eps / k} + 1} / \paren{ e^{\eps / k} - 1} \in O \paren{k / \eps}$, or equivalently,~$\cGap \in \Omega \paren{ \eps / k }$.
Our protocol relies on a randomizer, \ourRandomizer, with a~$\sqrt{k}$-better guarantee than this naive one.

\begin{theorem}[\ourRandomizer] \label{theorem: lower bound on c-gap}
    There is a randomizer~$\cM$, called \ourRandomizer,  that satisfies the above conditions with
    \begin{equation}
        \cGap \in \Omega \paren{ \eps / \sqrt{ k } }.
    \end{equation}
\end{theorem}

The design of \ourRandomizer is non-trivial, and we defer it to Section~\ref{sec: our randomizer}.
For now, we apply it to complete the design of our longitudinal data tracking protocol.

\vspace{3mm}
\noindent {\bf Client-Side Algorithm.}
The client-side algorithm is described in Algorithm~\ref{algo: client}.
Each user~$u$ first samples an integer~$h_u \in \IntSet{0}{\log d}$ uniformly at random and reports it to the server.
Then~$u$ initializes a zero vector~$\pmbomega_u$, of size~$L \doteq d / 2^{h_u} = \card{\dyadicIntervalSet{ h_u } }$ for the dyadic intervals with order~$h_u$.
They also initialize a randomizer~$\cM$, with parameters~$L$,~$k$ and~$\eps$.

The user reports to the server at time~$t \in [d]$ if and only if~$2^{h_u}$ divides~$t$.
In particular,
the user~$u$ computes partial sum~$\partialsum{u}{ \dyadicInterval{h}{j} }$ associated with the~$j \doteq t / 2^{h_u}$-th dyadic interval~$\dyadicInterval{h_u}{j}$,
with order $h_u$.
Note that each multiple of~$2^{h_u}$ is the first time that~$u$ has all data $\UsrDerivative{u}\bracket{ t' }$ needed to compute~$\partialsum{u}{\dyadicInterval{h_u}{j}}$, since
by definition:
$$
    \dyadicInterval{h_u}{j} 
        = 
        \set{ \paren{j - 1} \cdot 2^{h_u} + 1, \ldots, j \cdot 2^{h_u} }
        =
        \set{t - 2^{h_u} + 1, \ldots, t}\,,
$$
and~$\partialsum{u}{\dyadicInterval{h_u}{j}} = \sum_{ t' \in \dyadicInterval{h_u}{j} } \UsrDerivative{u}\bracket{ t' }$.
This partial sum is then perturbed by the randomizer $\cM^{(j)}$.
The perturbed value is stored 
in~$\pmbomega_u[j]$ and reported to the server.

\begin{algorithm}[!t]
    \caption{Client $\algoClient{}$}
    \label{algo: client}
    \begin{algorithmic}[1]
        \Require User Data~$\UsrDerivative{u}$; Privacy parameter $\eps$.
        \State\label{line: report h-u}{\bf report}~$h_u \uniffrom \IntSet{0}{\log d}$. \Comment{Send $h_u$ to the server}
        \State Set $L \leftarrow d / 2^{h_u}$ and~$\pmbomega_u \leftarrow \set{0}^{ L }$.
        \State $\cM.\init(L, k, \eps)$
        \For{each time $t \in [d]$}\Comment{User data~$\UsrDerivative{u}[t]$~arrives}
        \If{$2^{h_u}$ divides~$t$}
        \State $j \leftarrow t / 2^{h_u}$
        \State Compute the~$j^{(th)}$ partial sum~$\partialsum{u}{ \dyadicInterval{h}{j} }$.
        \State \label{line: algo 1 last line}
        {\bf report}~$\pmbomega_u [ j ] \leftarrow \cM^{(j)} \paren{ \partialsum{u}{ \dyadicInterval{h}{j} } }$
        \EndIf
        \EndFor
    \end{algorithmic}
\end{algorithm}

\noindent {\bf Privacy Guarantee.}
The privacy guarantee of Algorithm~\ref{algo: client} is an immediate consequence of \rndmzrPropI.

\begin{theorem} \label{theorem: client basic is private}
    $\algoClient$ is $\eps$-differentially private.
\end{theorem}

\begin{proof}[Proof of Theorem~\ref{theorem: client basic is private}]
    Consider a user~$u \in \bracket{n}$ with data~$\UsrDerivative{u}$.
    For each~$h \in \IntSet{0}{\log d}$, we have~$\P{ h_u = h}$ $= 1 / \paren{1 + \log d}$,
    and, conditioned on~$h_u = h$, we have $L = d / 2^h$.
    Via Inequality~(\ref{ineq: randomizer output prob range}), there exist~$\pMin, \pMax \in \paren{0, 1}$, s.t., $\pMax \le e^{\eps} \cdot \pMin$ and that for each~$w \in \{-1, 1\}^L$, it holds that
    $$
        \P{ \pmbomega_u[1] = w_1, \ldots, \pmbomega_u[L] = w_L \mid h_u = h } \in \bracket{ \pMin, \pMax}.
    $$
    Suppose that data~$\UsrDerivative{u}$ of user~$u$ is instead~$\UsrDerivative{u}'$, with other inputs unchanged.
    Running~$\algoClient$ based on~$\UsrDerivative{u}'$, we
    have~$h_u'$ and~$\pmbomega_u'$
    and thus~$\P{ h_u' = h }$ $= 1 / (1 + \log d)$ and that
    $$
        \P{ \pmbomega_u'[1] = w_1, \ldots, \pmbomega_u'[L] = w_L \mid h_u' = h } \in \bracket{ \pMin, \pMax}.
    $$
    Since~$\pMax \le e^{\eps} \cdot \pMin$, the ratio of outcome probabilities is
    $$
        e^{-\eps} \le 
        \frac
        {
        \P{ h_u = h, \pmbomega_u[1] = w_1, \ldots, \pmbomega_u[L] = w_L }
        }
        {
        \P{ h_u' = h, \pmbomega_u'[1] = w_1, \ldots, \pmbomega_u'[L] = w_L }
        }
        \le e^{\eps}
        \,.
    $$
    Since the output space of~$\algoClient$ is discrete, it satisfies Inequality~(\ref{ineq: def private algo}), and therefore is differentially private (Definition~\ref{def: Differential Privacy}). 
\end{proof}

\subsection{Server Side}
\label{subsec: server side}

\noindent {\bf Server-Side Algorithm.}
The server-side algorithm is described in Algorithm~\ref{algo: server}.
At the beginning of the algorithm, according to~$\algoClient$ (Algorithm~\ref{algo: client}, line~\ref{line: report h-u}), we assume that the server has received the~$n$ sampled orders~$h_1, \ldots, h_n$, one from each user.
The server partitions the users into~$1 + \log d$
subsets~$\cU_0, \ldots, \cU_{\log d}$, such that~$\cU_h \doteq \{ u \in \bracket{n} : h_u = h \}$ consists of the subset of users whose sampled orders equal~$h$.
Observe that each user~$u \in \cU_h$ reports their perturbed partial sums~$\pmbomega_u[1], \pmbomega_u[2]$, $\ldots, \pmbomega_u[d / 2^h]$, at times~$2^h, 2 \cdot 2^h, \ldots, d$, respectively.

\begin{algorithm}[!t]
    \caption{Server $\algoServer$}
    \label{algo: server}
    \begin{algorithmic}[1]
        \Require Reports $(h_u, \pmbomega_u)$ from the users.
        \State $\cU_h \leftarrow \{ u \in \bracket{n} : h_u = h \}, \forall h \in \IntSet{0}{\log d}$.
        \For{each time $t \in [d]$}
            \For{$h \in \IntSet{0}{\log d}$ s.t.~$2^h$ divides~$t$}
                \State $j \leftarrow t / 2^h$
                \State 
                \label{line: algo: server: update partial sum}
                $\EstPartialsum{ \dyadicInterval{h}{j} } \leftarrow \sum_{u \in \cU_h} (1 + \log d) \cdot  \paren{\cGap}^{-1} \cdot \pmbomega_u [ j ]$
            \EndFor
            \State 
            \label{line: algo: server: output estimate}
            {\bf output} $\EstPorpulationSum{t} \leftarrow \sum_{ \dyadicInterval{h}{j} \in \dyadicDecomposition{t} } \EstPartialsum{ \dyadicInterval{h}{j} }$
        \EndFor
    \end{algorithmic}
\end{algorithm}

Based on user reports, for each $u \in \bracket{n}, h \in \IntSet{0}{\log d}, j \in \bracket{d / 2^h}$, consider this estimator~$z_u[h, j]$ of~$\partialsum{u}{\dyadicInterval{h}{j}}$. If~$h_u = h$, then~$z_u[h, j] = (1 + \log d) \cdot  \paren{\cGap}^{-1} \cdot \pmbomega_u [ j ]$; otherwise,~$z_u[h, j] = 0$.
Since~$\pmbomega_u [j] \in \set{-1, 1}$, it holds that~$z_u[h, j] \in (1 + \log d) \cdot  \paren{\cGap}^{-1} \cdot \bracket{ - 1, 1 }$.
Further, the expectation of~$z_u[h, j]$ is given by 
$$
    \E{z_u[h, j]} = \P{ h_u = h } \cdot (1 + \log d) \cdot \E{  \paren{\cGap}^{-1} \cdot \pmbomega_u [ j ] }\,.
$$
Noting that~$\pmbomega_u [ j ] = \cM^{(j)} \paren{ \partialsum{u}{ \dyadicInterval{h}{j} } }$, and via Equation~(\ref{equa: unbias of each coordinate of M}), we see
\begin{equation}
    \E{z_u[h, j]} = \partialsum{u}{ \dyadicInterval{h}{j} }\,.
\end{equation}

Via linearity of expectation,~$\EstPartialsum{ \dyadicInterval{h}{j} } = \sum_{u \in \bracket{n}} z_u[h, j]$ is an unbiased estimator of~$\partialsum{\dyadicInterval{h}{j}}$.
Keeping only the non-zero terms, and replacing~$z_u[h, j]$ with its definition, we have
$$
    \EstPartialsum{ \dyadicInterval{h}{j} } = \sum_{u \in \cU_h} (1 + \log d) \cdot  \paren{\cGap}^{-1} \cdot \pmbomega_u [ j ]\,,
$$
which justifies the update rule of~$\EstPartialsum{ \dyadicInterval{h}{j} }$ (Algorithm~\ref{algo: server}, line~\ref{line: algo: server: update partial sum}).
Finally, replacing~$\partialsum{\dyadicInterval{h}{j}}$ with~$\EstPartialsum{ \dyadicInterval{h}{j} }$ in Equation~(\ref{equa: decomposition of porpulatioin sum}) gives an unbiased estimator,~$\EstPorpulationSum{t} = \sum_{ \dyadicInterval{h}{j} \in \dyadicDecomposition{t} } \EstPartialsum{ \dyadicInterval{h}{j} }$, of~$\PorpulationSum{t}$ (Algorithm~\ref{algo: server}, line~\ref{line: algo: server: output estimate}).

\noindent {\bf Utility Guarantee.}
We conclude that~$\EstPorpulationSum{t}$ is an unbiased estimator of~$\PorpulationSum{t}$; how often is it accurate?

\begin{lemma} \label{lemma: error of single estimator for basic private algorithm}
    For each~$t \in [d]$, assuming that $\eps \le 1$ and that $\sqrt{n} \ge \paren{\cGap}^{-1} \cdot \log d \cdot \sqrt{ \ln \paren{1 / \beta'} }$, with probability at least $1 - \beta'$, 
    the estimates~$\EstPorpulationSum{t}$ computed as in Algorithm~\ref{algo: server} satisfy
    \vspace{-1mm}
    $$
        | \EstPorpulationSum{t} - \PorpulationSum{t} | \in  O \left( \paren{\cGap}^{-1} \cdot \log d \cdot \sqrt{  n  \cdot \ln \frac{1}{\beta'} } \right)
    $$
    \vspace{-2mm}
\end{lemma}

We first prove our main Theorem~\ref{theorem: property of further improved private algorithm} based on Theorems~\ref{theorem: lower bound on c-gap} and Lemma~\ref{lemma: error of single estimator for basic private algorithm}, then justify Lemma~\ref{lemma: error of single estimator for basic private algorithm}.

\begin{proof}[Proof of Theorem~\ref{theorem: property of further improved private algorithm}]
    Via Theorem~\ref{theorem: lower bound on c-gap}, we know that~$\cGap \in \Omega \paren{ \eps / \sqrt{k} }$.
    Therefore, $\paren{\cGap}^{-1} \in O \paren{ \sqrt{k} / \eps }$.
    Applying Lemma~\ref{lemma: error of single estimator for basic private algorithm} with~$\beta' = \beta / d$, and via union bound, we conclude that, with probability at least~$1 - \beta$,
    \vspace{-1mm}
    $$
        \max_{t \in \bracket{d} } | \EstPorpulationSum{t} - \PorpulationSum{t} | 
        \in  
        O \left( \frac{ \log d } {\eps} \cdot \sqrt{  k \cdot n  \cdot \ln \frac{d}{\beta} }  \right).
    $$
    \vspace{-2mm}
\end{proof}

\vspace{-2mm}
\begin{proof}[Proof of Lemma~\ref{lemma: error of single estimator for basic private algorithm}]
    First, rewrite
    $$
        \begin{aligned}
            \EstPorpulationSum{t}
             & = \sum_{ \dyadicInterval{h}{j} \in \dyadicDecomposition{t} } \EstPartialsum{ \dyadicInterval{h}{j} }
            = \sum_{ \dyadicInterval{h}{j}  \in \dyadicDecomposition{t} } \sum_{u \in \bracket{n} } z_u [h, j]\,.                
        \end{aligned}
    $$

    \noindent For each $u \in \bracket{n}$, define $Y_u = \sum_{ \dyadicInterval{h}{j} \in \dyadicDecomposition{t} }  z_u [h, j]$.
    Exchanging the order of summation gives~$\EstPorpulationSum{t} = \sum_{u \in \bracket{n} } Y_u$.
    Clearly the~$Y_u$ are independent.
    We \emph{claim} (proven below): $Y_u \in (1 + \log d) \cdot  \paren{\cGap}^ \cdot [-1, 1]$.
    Then applying  Hoeffding's Inequality (Corollary~\ref{corollary: confident interval of hoeffding}), and the fact that $\E{ \EstPorpulationSum{t} } = \PorpulationSum{t}$, we see that with probability at most~$\beta'$,
    \begin{equation}
        \card{ \EstPorpulationSum{t} - \PorpulationSum{t} }
        \ge (1 + \log d) \cdot  \paren{\cGap}^{-1} \cdot \sqrt{ 2 n \cdot \ln \frac{2}{\beta'} }\,.
    \end{equation}

    We prove the claim, that $Y_u \in (1 + \log d) \cdot  \paren{\cGap}^{-1} \cdot [-1, 1]$.
    By Fact~\ref{fact: nodes be in different level}, dyadic intervals in $\dyadicDecomposition{t}$ have distinct orders.
    By definition,~$z_u[h, j]$ is non-zero exactly when~$h_u = h$,
    so among all $\dyadicInterval{h}{j}  \in \dyadicDecomposition{t}$, there is at most one non-zero $z_u[h, j]$. Its value either $-(1 + \log d) \cdot  \paren{\cGap}^{-1}$ or $(1 + \log d) \cdot  \paren{\cGap}^{-1}$.
\end{proof}

\section{Randomizer} 
\label{sec: our randomizer}

In this section, we present a randomizer, the \ourRandomizer, denoted briefly by~$\cM$, that satisfies 
Theorem~\ref{theorem: lower bound on c-gap}. 
The \ourRandomizer is based on two techniques, \emph{composition for randomized responses} (for non-zero coordinates) and \emph{pre-computation for the composition}.

\subsection{Overview}

The input to~$\cM$ is a $k$-sparse sequence $v = \paren{v_1, \ldots, v_L} \in \{-1, 0, 1\}^L$, and the output is sequence~$\cM^{(1)} (v_1), \ldots, \cM^{(L)} (v_L) \in \{-1, 1\}^L$.
Denote the support of~$v$ as~$\supp{ v } \doteq \set{ j \in   [L] : v_j \neq 0 }$.

\vspace{2mm}
\noindent {\bf Zero Coordinates.}
For each~$j \notin \supp{v}$,~$\cM^{ (j) } (v_j)$ outputs~$-1$ and $1$ uniformly at random. Hence \rndmzrPropIII is trivially satisfied. 

\vspace{2mm}
\noindent {\bf Non-Zero Coordinates.}
The presentation of \ourRandomizer to handle non-zero coordinates follows three steps.

\begin{itemize}
    \item {\it Offline Input with Fixed Support Size.} 
    First we assume that all coordinates of~$v$ are inputted to~$\cM$ simultaneously, and that the~$v$ contains exactly~$k$ non-zero coordinates. 
    
    \item {\it Online Input with Fixed Support Size.} 
    We convert the protocol to online, where each coordinate~$v_j$ arrives one by one, and~$\cM$ outputs the perturbed value of~$v_j$ immediately. 
    
    \item {\it Online Input with Bounded Support Size.} 
    We show that
    \emph{Online Input with Fixed Support Size} protocol provides the same guarantees even when~$v$ has support less than~$k$.
\end{itemize}

\noindent For each step, we need to show that the~$\cM$ constructed satisfies \rndmzrPropI, \rndmzrPropII with~$\cGap \in \Omega \paren{ \eps / \sqrt{k}}$.

\subsection{Offline Input with Fixed Support Size}

By assumption, $v$ contains~$k$ non-zero coordinates.
$\cM$ invokes a subroutine,~$\tilde \cR : \set{-1, 1}^k \rightarrow \set{-1, 1}^k$ to perturb these coordinates. 
$\tilde \cR$ is a \emph{composed randomizer}: instead of perturbing each coordinate independently, it adds correlated noise to the coordinates. 
The building block of~$\tilde{\cR}$ is a basic randomizer~$\cR$. 
The pseudo-codes of both~$\cR$ and~$\tilde{\cR}$ are described in Algorithm~\ref{algo: our randomizer}.

\vspace{1mm}
\noindent{\bf Basic Randomizer~$\cR$}~\citep{W65}.
For each~$\zeta \in \set{-1, 1}$.
\begin{equation}
    \label{equa: def of basic randomizer}
    \cR(\zeta) = \begin{cases}
        \zeta,      & \text{w.p.}\, e^{ \tildeeps } / \paren{ e^{ \tildeeps } + 1 } \\
        -\zeta,     & \text{w.p.}\, 1 / \paren{ e^{ \tildeeps } + 1 }
    \end{cases}, 
\end{equation} 
where~$\tildeeps$ depends on~$\eps$, defined later. 

\vspace{1mm}
\noindent{\bf Composed Randomizer~$\tilde \cR$.}
For an input~$b \in \set{-1, 1}^k$, it first applies~$\cR$ independently to each coordinate of $b$.
Denote the result by~$b' = \cR(b) \doteq \paren{\cR(b_1), \ldots, \cR(b_k) }$. 
Denote $p \doteq 1 / ( e^{ \tildeeps } + 1)$.
In expectation, 
$\norm{b' - b}_0 = kp$.
Let
\begin{equation}
    \label{equa: setting of Lb and Ub}
    \NewLb \doteq k p - 2 \sqrt{k},\,\quad
    \NewUb \doteq  \frac{k}{ \tildeeps } \cdot \ln \frac{2 e^{\tildeeps} }{ e^{\tildeeps} + 1 }\,.
\end{equation}

\noindent Next, define the set of sequences whose~$\ell_0$-distance (Hamming distance) to $b$ is within~$\IntSet{\NewLb}{\NewUb}$.
\begin{definition}[Annulus] Given $b \in \set{-1, 1}^k$, denote
    $$
        \Annulus{b}
        \doteq
        \set{ s  \in \set{-1, 1}^k : \norm{ b - s }_0 \in     \IntSet{\NewLb}{\NewUb}
        }.
    $$
\end{definition}

If~$b' \notin \Annulus{b}$,~$\tilde{\cR}$ replaces it with a uniform sample from~$\set{-1, 1}^k \setminus \Annulus{b}$. 
Finally,~$\tilde{\cR}$ outputs~$b'$. 

Randomizer~$\tilde \cR$ has the following privacy and utility guarantees, the analysis of which is deferred to Section~\ref{subsec: analysis of Lemmas of Composed Randomizer}. 

\begin{lemma}
    \label{lemma: probability of modified composed randomizer}
    Suppose that~$\eps \le 1$ and let~$\tildeeps = \eps / \paren{ 5 \sqrt{k} }$.
    There exist~$\pMin', \pMax' \in \paren{0, 1}$ with~$\pMax' \le e^\eps \cdot \pMin'$, such that for each input sequence~$b \in \set{-1, 1}^k$, and each sequence~$s \in \set{-1, 1}^k$,
    \begin{equation}
        \label{ineq: output range of compose randomizer}
        \P{ \tilde \cR(b) = s } \in \bracket{ \pMin', \pMax' }\,. 
    \end{equation}
    for $\tilde \cR(b)$ as defined in Algorithm~\ref{algo: our randomizer}.
\end{lemma}

\begin{lemma}
    \label{lemma: lower bound of csvr}
    For a given input~$b$, denote~$\tilde{b} \doteq \tilde{\cR}(b)$ the output of~$\tilde{\cR}$.
    There exits some~$\cGap \in \Omega \paren{\tildeeps} = \Omega \paren{ \eps / \sqrt{k} }$, such that for each input~$b \in \set{-1, 1}^k$, and for each~$i \in \bracket{k}$,
    \begin{equation}
        \P{ \tilde{b}_i = b_i } - \P{ \tilde{b}_i = -b_i } = \cGap \in \Omega \paren{ \eps / \sqrt{k} }\,. 
    \end{equation}
\end{lemma}

\begin{algorithm}[!t]
    \caption{\ourRandomizer\,$\cM$, describing $\tilde \cR$}
    \label{algo: our randomizer}

    \begin{algorithmic}[1]
        \Procedure{Basic Randomizer~$\cR$}{$\zeta$}
        \Require Value $\zeta \in \set{-1, 1 }$.
        \State {\bf return} $-\zeta$ w.p.~$1 / \paren{ e^{ \tildeeps } + 1 }$ and $\zeta$ w.p.~$e^{ \tildeeps } / \paren{ e^{ \tildeeps } + 1 }$.
        \EndProcedure
        
        \Statex \vspace{-3mm}
        \Procedure{Composed Randomizer~$\tilde \cR $}{$b$}
        \Require Vector $b \in \set{-1, 1}^k$.
        \State Sample $b' \leftarrow \paren{ \cR ( b_1 ), \ldots, \cR ( b_k ) }$.
        \label{algo line: 1 algo:Approximate Composed Algorithm}
        \If{ $b' \notin \Annulus{b}$ }
        \State $b' \uniffrom \set{-1, 1}^k \setminus \Annulus{b}$
        \label{algo line: 3 algo:Approximate Composed Algorithm}
        \EndIf
        \State {\bf return} $b'$
        \EndProcedure

        \Statex \vspace{-3mm}
        \Procedure{$\cM.\init$}{$L, k, \eps$}
        \Require Input length~$L$; Support size~$k$; Privacy parameter~$\eps$
        \State Set~$\tildeeps \leftarrow \eps / \paren{5 \sqrt{k}}$.
        \State Set~$\tilde{b} \leftarrow \tilde{\cR}( 1^k)$.
        \State Set~$\nnz \leftarrow 0$.
        \EndProcedure

        \Statex \vspace{-3mm}
        \Procedure{$\cM^{ (j) }$}{$v_j$}\Comment{$j = 1, 2, \ldots, L$}
        \Require Value $v_j \in \set{-1, 0, 1 }$.
        \If{$v_j \neq 0$}
        \State $\nnz \leftarrow \nnz + 1$.
        \State {\bf return} $v_j \cdot \tilde{b}_{\nnz}$.
        \Else
        \State {\bf return} $-1$ or $+1$ uniformly at random.
        \EndIf
        \EndProcedure
    \end{algorithmic}
\end{algorithm}

\noindent{\bf Analysis.}
Since~$\cM$ perturbs the non-zero coordinates of input~$v$ with~$\tilde{\cR}$, \rndmzrPropII is satisfied based on the Lemma~\ref{lemma: lower bound of csvr}.
Hence, it is left to verify \rndmzrPropI. 
Consider an arbitrary sequence~$w_1, \ldots, w_L \in \set{-1, 1}$. 
The event
$$
    E: M^{ (1) } (v_1) = w_1, \ldots, M^{ (L) } (v_L) = w_L
$$ 
can be decomposed into two events: 
\begin{enumerate}
    \item~$E_1:$ $\forall j \notin \supp{ v }$, $\cM^{ (j) } ( v_{ j } ) = w_{ j }$.
    \item~$E_2:$ $\forall j \in \supp{ v }$, $\cM^{ (j) } ( v_{ j } ) = w_{ j }$.
\end{enumerate}

We have~$\P{E_1} = 2^{- \paren{L - k}}$ as the~$\cM^{ (j) } ( v_{ j } )$ values are independent in $\{-1,+1\}$ for~$j \notin \supp{ v }$. 
To bound~$\P{E_2}$, denote the indices in~$\supp{v}$ as $j_1 <  \cdots <  j_k$.
Since~$\cM$ perturbs the non-zero coordinates with~$\tilde{\cR}$, applying Lemma~\ref{lemma: probability of modified composed randomizer} with input~$b = \paren{ v_{j_1}, \ldots, v_{j_k} }$ and output~$s = \paren{ w_{j_1}, \ldots, w_{j_k} }$,
it holds that~$\P{ E_2 } = \P{ \tilde{\cR} (b) = s } \in \bracket{ \pMin', \pMax' }$. 
Since~$E_1$ and~$E_2$ are independent, we have
\begin{equation}
     \label{ineq: output range of our randomizer}
    \P{E} = \P{E_1} \cdot \P{E_2} \in 2^{- \paren{L - k}} \cdot \bracket{ \pMin', \pMax' }\,,
\end{equation}
which proves \rndmzrPropI.

\subsection{Online Input with Fixed Support Size}
\label{sebsec: online input with fixed support}
In this step, we modify~$\cM$ into an online algorithm, where each coordinate~$v_j$ arrives one by one, and~$\cM$ is required to perturb and output each coordinate immediately after its arrival. 
Since~$\cM$ perturbs the zero coordinates independently, 
we only need to take care of the non-zero coordinates. 
We develop a new pre-computation technique to generate the correlated noises for the non-zero elements in the initialization phase. 
With these noises, we perturb the non-zero coordinates as they arrive.  
This motivates the design of $\cM.\init \paren{L, k, \eps}$ and $\cM^{ (j) } \paren{v_j}$, $j \in \bracket{L}$, whose pseudo-codes are in Algorithm~\ref{algo: our randomizer}.

\vspace{2mm}
\noindent{\bf $\cM.\init \paren{L, k, \eps}$.}
The procedure takes as parameters~$L$, the size of the input sequence;~$k$, the maximum number of non-zero elements in the input sequence; and~$\eps$, the privacy parameter.
It sets~$\tildeeps$ as $\eps / \paren{5 \
\sqrt{k}}$, and invokes the composed randomizer~$\tilde \cR$ with the vector~$1^k$ consisting of all ones.
The returned sequence is kept as a vector~$\tilde{b}$.
Finally, the procedure creates a variable~$\nnz$ with value~$0$; the variable~$\nnz$ is shared by other procedures, to keep track of the support of the input sequence received so far.

\vspace{2mm}
\noindent{\bf $\cM^{ (j) } \paren{v_j}$.}
For all $j \in \bracket{L}$,~$\cM^{ (j) } \paren{v_j}$ shares the same pseudo-codes. If~$v_j = 0$, the procedure outputs~$-1$ or~$+1$, uniformly at random. 
Otherwise, it increases~$\nnz$ by~$1$, indicating that~$v_j$ is the~$\paren{ \nnz }^{ (th) }$ non-zero entry processed by the procedure.
Then it outputs the value of~$v_j$ multiplied by $\tilde{b}_{\nnz}$.

\vspace{2mm}
\noindent{\bf Analysis.}
We need to prove that \rndmzrPropI and \rndmzrPropII hold.
Consider the events~$E_1$ and~$E_2$ defined as before. 
We see that~$\P{E_1}$ remains the same.
Denote the indices in~$\supp{v}$ as $j_1 <  \cdots <  j_k$.
Based on our modification, for each~$i \in \bracket{k}$, we have 
$\cM^{ (j_i) } ( v_{j_i} ) = \tilde{b}_i \cdot v_{j_i}$.
Hence, event~$E_2$ holds only if for each~$i \in \bracket{k}$, 
$\tilde{b}_i \cdot v_{j_i} = w_{j_i}$, equivalently,~$\tilde{b}_i = w_{j_i} / v_{j_i}$. 
Since~$\tilde{b} = \tilde{\cR}(1^k)$, applying Lemma~\ref{lemma: probability of modified composed randomizer} with~$s = \paren{ w_{j_1} / v_{j_1}, \ldots, w_{j_k} / v_{j_k}}$,
it holds that
\vspace{-1mm}
$$
    \P{ E_2 } = \P{ \tilde{b} = s } = \P{ \tilde{\cR} (1^k) = s } \in \bracket{ \pMin', \pMax' }\,.
$$ 
It follows that Inequality~(\ref{ineq: output range of our randomizer}) still holds, and therefore \rndmzrPropI.  

Finally, observe that~$\cM^{ (j_i) } ( v_{j_i} ) = v_{j_i}$ if~$ \tilde{b}_i = 1$; and~$\cM^{ (j_i) } ( v_{j_i} ) = -v_{j_i}$ if~$ \tilde{b}_i = -1$. 
Since~$\tilde{b} = \tilde \cR \paren{1^k}$, we have $\P{ \tilde{b}_i = 1 } - \P{ \tilde{b}_i = -1 } = \cGap$.
Therefore, \rndmzrPropII holds with~$\cGap \in \Omega \paren{ \eps / \sqrt{k}}$.

\subsection{Online Input with Bounded Support Size}
In this step, we relax the constraint that vector~$v$ contains exactly~$k$ non-zero coordinates. 
In particular, we show that \rndmzrPropI and \rndmzrPropII still hold, if the same protocol discussed in Section~\ref{sebsec: online input with fixed support} acts on input with support less than~$k$.   
In the case~$\card{\supp{v}} = k$, each bit of the pre-generated vector~$\tilde{b} = \tilde{\cR}(1^k)$ is used to multiply some non-zero coordinate of~$v$.
In the case~$\card{\supp{v}} < k$, however, only the first~$\card{\supp{v}}$ bits of~$\tilde{b}$ are used.

\vspace{2mm}
\noindent{\bf Analysis.}
With a similar argument as previous section, \rndmzrPropII holds with~$\cGap \in \Omega \paren{ \eps / \sqrt{k}}$. 
To verify \rndmzrPropI, we prove that Inequality~(\ref{ineq: output range of our randomizer}) still holds. 
Consider the events~$E_1$ and~$E_2$ defined as before. 
We have~$\P{E_1} = 2^{ -(L - \card{\supp{v}} ) }$ as the~$\cM^{ (j) } ( v_{ j } )$ variables are independent in $\{-1, 1\}$ random variables for~$j \notin \supp{ v }$. 
The event~$E_2$ happens if only for each $i \in \bracket{\card{ \supp{ v } }}$, it holds that $\tilde{b}_i = { w_{ j_i }  } / { v_{ j_i } }$.
Let~$\targetOutputs$ be the subset of sequence~$s \in \set{-1, 1}^k$ which satisfies~$s_i = { w_{ j_i }  } / { v_{ j_i } }$, for each~$i \in \bracket{\card{ \supp{ v } }}$.
There are~$2^{k - \card{\supp{v}} }$ such possible sequences, each being outputted by~$\tilde{\cR} \paren{ 1^k }$ with probability between~$\bracket{ \pMin', \pMax' }$ (by Lemma~\ref{lemma: probability of modified composed randomizer}).
Therefore, 
$$
    \P{ E_2 } = \P{ \tilde{b} \in \targetOutputs } \in 2^{k - \card{\supp{v}} } \cdot \bracket{ \pMin', \pMax' }. 
$$
Multiplying it by~$\P{E_1} = 2^{ -(L - \card{\supp{v}} ) }$ proves Inequality~(\ref{ineq: output range of our randomizer}).

\subsection{Sketch Proofs}
\label{subsec: analysis of Lemmas of Composed Randomizer}

We conclude the technical presentation with outlines of our proofs for Lemma~\ref{lemma: probability of modified composed randomizer} and~\ref{lemma: lower bound of csvr}. 
The complete proofs for Lemma~\ref{lemma: probability of modified composed randomizer} and~\ref{lemma: lower bound of csvr} are included in Appendix~\ref{appendix: subsubsec lemma: probability of modified composed randomizer} and~\ref{appendix: subsubsec lemma: lower bound of csvr}, respectively.

We use the following notation.
As before, denote $p = 1 / ( e^{ \tildeeps } + 1)$.
Note that~$1 - p = e^{\tildeeps} p$.
For each $i \in \IntSet{0}{k}$, define
$$
    g(i) \doteq p^i (1 - p)^{k - i}
    = p^k \cdot e^{ \tildeeps \cdot (k - i) }\,.
$$
We see that~$g$ is a decreasing function with respect to~$i$.
For each~$b \in \set{-1, 1}^k$, denote the result of applying~$\cR$ independently to each of its coordinates as~$\cR(b) \doteq \paren{\cR(b_1), \ldots, \cR(b_k) }$. 
For each~$s \in \set{ -1, 1 }^k$, it is easy to verify that $\P{ \cR(b) = s } = g \paren{ \norm{ s - b }_0 }$,
where~$\norm{ \cdot }_0$ is the~$\ell_0$ norm, and therefore~$\norm{ b - s }_0$ is the number of coordinates in which~$s$ differs from~$b$.
Since in expectation, $\norm{ \cR(b) - b} = kp$,
we define
\begin{align*}
    \pavg \doteq g(kp) = p^{kp} (1 - p)^{k - kp}  
    = p^k \cdot e^{ \tildeeps \cdot (k - kp) }.
\end{align*}

\noindent {\bf Proof Outline for Lemma~\ref{lemma: probability of modified composed randomizer}.}
    Let~$b \in \set{-1, 1}^k$ be an input to~$\tilde{\cR}$. 
    Recall that~$\tilde{\cR}$ first computes~$b' = \cR(b) \doteq \paren{\cR(b_1), \ldots, \cR(b_k) }$. 
    If~$b' \in \Annulus{b}$,~$\tilde{\cR}$ outputs it directly. 
    Otherwise,~$\tilde{\cR}$ outputs a uniform sample from~$\set{-1, 1}^k \setminus \Annulus{b}$. 
    At a high level, $\Annulus{b}$ consists of the~$s \in \set{-1, 1}^k$ for which~$\P{\cR(b) = s}$ is close to~$\pavg$, and $\tilde{\cR}$ keeps their output probabilities.
    On the other hand, $\set{-1, 1}^k \setminus \Annulus{b}$ consists of the~$s$ for which~$\P{\cR(b) = s}$ is much higher or lower than~$\pavg$, and $\tilde{\cR}$ averages their output probabilities by uniform sampling. 
    We will show that in both cases, $\P{ \tilde{\cR} (b) = s} \approx \pavg$.
    Formally, 
    \begin{align}
        \label{ineq: inner probability of modified composed randomizer}
        \P{ \tilde{\cR} (b) = s  } 
        &\in \left[
            1 / 2^k , e^{ 2 \tildeeps \sqrt{ k } } \cdot
            \pavg
            \right],\, \quad \forall s \in \Annulus{b}\,,
        \\
        \label{ineq: outter probability of modified composed randomizer}
        \P{ \tilde{\cR} ( b ) = s  } 
        &\in \left[
            e^{-3 \tildeeps \sqrt{ k } } \cdot \pavg,
            1 / 2^k
            \right],\, \quad \forall s \notin \Annulus{b}\,.
    \end{align}
    Let~$\pMin' = e^{-3 \tildeeps \sqrt{ k } } \cdot \pavg$ and~$\pMax' = e^{ 2 \tildeeps \sqrt{ k } } \cdot \pavg$.
    Since~$\tildeeps = \eps / \paren{ 5 \sqrt{k}}$, $\pMax' = e^\eps \cdot \pMin$.
    Combing Inequality~(\ref{ineq: inner probability of modified composed randomizer}) and Inequality~(\ref{ineq: outter probability of modified composed randomizer}), we know that for all $s \in \set{-1, 1}^k$, $\P{ \tilde{\cR} ( b ) = s  } \in \bracket{\pMin', \pMax'}$, which proves Lemma~\ref{lemma: probability of modified composed randomizer}.
    We now prove these two inequalities separately.

    \vspace{1mm}
    \noindent {\bf Bounding~$ \P{ \tilde{\cR} (b) = s  }$ for $s \in \Annulus{b}$.}
    The design of~$\tilde{\cR}$ ensures that 
    for each~$s \in \Annulus{b}$, we have~$\P{\tilde{\cR} (b) = s  } = \P{ \cR(b) = s  } = g( \norm{ b - s }_0 )$.
    Our choices of~$\NewUb$ and~$\NewLb$ in Equation~(\ref{equa: setting of Lb and Ub}) guarantee that~$\norm{b - s}_0$ is close to~$kp$ if~$s \in \Annulus{b}$, and therefore
    ~$\P{\tilde{\cR} (b) = s }$ is close to~$\pavg$.
    Next, we justify these choices. 
    
    \noindent {\it Choice of~$\NewUb$.} 
    We set~$\NewUb = \frac{k}{ \tildeeps } \cdot \ln \frac{2 e^{\tildeeps} }{ e^{\tildeeps} + 1 }$ so that~$g(\NewUb) = 2^{-k}$.
    Such~$\NewUb$ is close to~$kp$; 
    indeed, we have~$\NewUb \in [kp, k / 2]$.
    As~$g$ is a decreasing function, this can be proven by showing that
    \begin{equation}
        g(kp) \ge 2^{-k} = g(\NewUb) \ge g( k / 2)\,.
    \end{equation}
    Since~$\norm{b - s}_0 \le \NewUb$ for each~$s \in \Annulus{b}$, we have
    \begin{equation}
        \P{\tilde{\cR} (b) = s} \ge g(\NewUb) = 2^{-k}\,.
    \end{equation}
    As will be discussed, this property also plays an important role in upper bounding~$\P{\tilde{\cR} (b) = s}$ for~$s \notin \Annulus{b}$.
    
    \noindent {\it Choice of~$\NewLb$.} 
    We pick~$\NewLb = k p - 2 \sqrt{k}$ so that it is not much smaller than~$kp$, and that~$g( \NewLb ) = e^{ 2 \tildeeps \sqrt{ k } } \cdot \pavg$.
    Since~$\norm{b - s}_0 \ge \NewLb$ for each~$s \in \Annulus{b}$, 
    \begin{equation}
        \P{\tilde{\cR} (b) = s  } = g( \norm{ b - s }_0 ) \le  g( \NewLb ) = e^{ 2 \tildeeps \sqrt{ k } } \cdot \pavg\,.
    \end{equation}

    \vspace{2mm}
    \noindent {\bf Bounding~$\P{ \tilde{\cR} (b) = s  }$ for $s \notin \Annulus{b}$.}
    We discuss first the upper bound for this case, which is the easy part. 
    
    \noindent {\it Upper Bound.}  
    As discussed, our choice of~$\NewUb$ guarantees that each~$s \in \Annulus{b}$ is assigned with output probability at least~$2^{-k}$. 
    Since there are $2^k$ elements in the output space~$\set{-1, 1}^k$ of~$\tilde{\cR} (b)$, and since $\P{ \tilde{\cR} (b) = s  }$ equals a common probability for each~$s \notin \Annulus{b}$, we have
    $$
        \P{ \tilde{\cR} (b) = s } \le 2^{-k},\quad  \forall s \notin \Annulus{b}. 
    $$

    \vspace{1mm}
    \noindent {\it Lower Bound.}
    Let $\overline{ \IntSet{\NewLb}{\NewUb} } \doteq \IntSet{0}{k} \setminus \IntSet{\NewLb}{\NewUb}$, and~$\cR(b) = \paren{\cR(b_1), \ldots, \cR(b_k) }$.
    First, observe that
    $$
        \begin{array}{c}
             \P{ \cR(b) \notin \Annulus{b} } = \sum_{ i \in \overline{ \IntSet{\NewLb}{\NewUb} } } \binom{k}{i} g(i).
        \end{array}
    $$
    As~$\tilde{\cR}$ assigns equal probability to each~$s \notin \Annulus{b}$,  
    it holds that 
    \begin{equation}
        \PrOut^* \doteq \P{ \tilde{\cR} (b) = s } = 
        \frac{
            \sum_{ i \in \overline{ \IntSet{\NewLb}{\NewUb} } } \binom{k}{i} g(i)
        }
        {
            \sum_{ i \in \overline{ \IntSet{\NewLb}{\NewUb} } } \binom{k}{i}
        }\,.
    \end{equation}
    To lower bound~$\PrOut^*$, we will partition~$\overline{\IntSet{\NewLb}{\NewUb}}$ into three subsets, which consist of the~$i$ for which~$g(i)$ is significantly higher, lightly lower than, and significantly lower than~$\pavg$, respectively. 
    
    In particular, since $\NewLb < k / 2$ and~$\NewUb \le k / 2$, it holds that $k - \NewLb > \NewUb$.
    Hence, we can partition the set~$\overline{\IntSet{\NewLb}{\NewUb}}$ into three subsets~$\IntSet{0}{\NewLb - 1}$, $\IntSet{\NewUb + 1}{k - \NewLb}$ and~$\IntSet{k - \NewLb + 1}{k}$.
    Note that the size of the first subset equals the size of the third one.
    Correspondingly, we can decompose the numerator and denominator of $\PrOut^*$ into three parts:
    \vspace{-1mm}
    \begin{equation*}
        \numerator{1} \doteq \sum_{ i = 0 }^{ \NewLb - 1 } \binom{k}{i} g(i),
        \numerator{2} \doteq \sum_{ i = \NewUb + 1 }^{ k - \NewLb} \binom{k}{i} g(i),
        \numerator{3} \doteq \sum_{ i = k - \NewLb + 1 }^{ k} \binom{k}{i} g(i)\,.
    \end{equation*}
    \vspace{-3mm}
    \begin{align*}
        \denominator{1} \doteq \sum_{ i = 0 }^{ \NewLb - 1 } \binom{k}{i},\,
        \denominator{2} \doteq \sum_{ i = \NewUb + 1 }^{ k - \NewLb} \binom{k}{i},\,
        \denominator{3} \doteq \sum_{ i = k - \NewLb + 1 }^{ k} \binom{k}{i}\,.
    \end{align*}
    We lower bound the ratios of~${\numerator{2}} / {\denominator{2}}$ and~$\paren{\numerator{1} + \numerator{3}} / \paren{\denominator{1} + \denominator{3}}$ separately.   
    
    \noindent {\it Lower bounding~${\numerator{2}} / {\denominator{2}} \ge \pavg \cdot e^{- \tildeeps \cdot \paren{ k(1 - 2p) + 2 \sqrt{k} } }$.}
    This follows from that~$g(i)$ is a decreasing function, and therefore the smallest~$g(i)$ in the summands of~$\numerator{2}$ is lower bounded by~$ g(k - \NewLb) = \pavg \cdot e^{- \tildeeps \cdot \paren{ k(1 - 2p) + 2 \sqrt{k} } }$.
    
    \noindent {\it Lower bounding~$\paren{\numerator{1} + \numerator{3}} / \paren{\denominator{1} + \denominator{3}} \ge \pavg \cdot e^{- \tildeeps \cdot k (1 - 2p)  }$.}
    Observe that for each $i \in \IntSet{0}{\NewLb - 1}$, we have $k - i \in \IntSet{k - \NewLb + 1}{k}$.
    The lower bound follows by pairing up each summand~$\binom{k}{i} g(i)$ in $\numerator{1}$ with $\binom{k}{k - i} g(k - i)$ in~$\numerator{3}$, and by proving that~$g(i) + g(k - i) \ge \pavg \cdot 2 \cdot e^{- \tildeeps \cdot k (1 - 2p)  }$,
    details in the Appendix.

    \vspace{1mm}
    \noindent {\it Putting Together.}
    Based on the lower bounds of~${\numerator{2}} / {\denominator{2}}$ and~$\paren{\numerator{1} + \numerator{3}}$ $/\paren{\denominator{1} + \denominator{3}}$, we have 
    \begin{align}
        \PrOut^*
        = \frac{ \numerator{1} + \numerator{2} +  \numerator{3} }{ \denominator{1} + \denominator{2} + \denominator{3} }
        \ge \pavg \cdot e^{- \tildeeps \cdot \paren{ k(1 - 2p) + 2 \sqrt{k} } }\,.
    \end{align}
    Since~$p = 1 / ( e^{ \tildeeps } + 1)$, we have $(1 - 2p) = ( e^{\tildeeps} - 1 ) / ( e^{\tildeeps} + 1  )$. 
    Via the inequality $\paren{e^x - 1} / \paren{e^x + 1} \le x / 2$ for all $x \ge 0$, we have $k (1 - 2p) \le k \cdot \tildeeps / 2$.
    Via the assumption that~$\eps \le 1$, we get~$\tildeeps = \eps / \paren{ 5 \sqrt{k} } \le 1 / \sqrt{k}$, and~$k (1 - 2p) \le k \cdot \tildeeps / 2 \le \sqrt{k}$.
    Therefore,
    $
        \PrOut^*
        \ge
        e^{- \tildeeps \cdot \paren{ 3 \sqrt{k} } }  \cdot \pavg
    $.

\hfill $\square$

\vspace{3mm}
\noindent {\bf Proof Outline for Lemma~\ref{lemma: lower bound of csvr}.}
    By symmetry of~$\tilde{\cR}$, it suffices to prove Lemma~\ref{lemma: lower bound of csvr} for the first coordinate of~$\tilde{b}$.
    We first prove that 
    \vspace{-3mm}
    \begin{align*}
        \P{ \tilde{b}_1 = b_1 } 
        &=\sum_{ i = \NewLb }^{  \NewUb } \binom{ k }{ i } g(i) \frac{k - i}{k}
        + \PrOut^* 
        \sum_{ i \in \overline{ \IntSet{\NewLb}{\NewUb} } } \binom{ k  }{ i  } \frac{k - i}{k},\, \\
        \P{ \tilde{b}_1 = -b_1 } 
        &= \sum_{i = \NewLb }^{  \NewUb } \binom{ k }{ i } g(i) \frac{ i }{ k }
        + \PrOut^* 
        \sum_{ i \in \overline{ \IntSet{\NewLb}{\NewUb} } } \binom{ k  }{ i  } \frac{i}{k}.
    \end{align*}
    Since~$\cGap = \P{ \tilde{b}_1 = b_1 } - \P{ \tilde{b}_1 = -b_1 }$, we have 
    \begin{equation*}
        \cGap = \sum_{ i = \NewLb }^{ \NewUb } \binom{k}{i} g(i) \frac{k - 2i}{k}
        + \PrOut^* \sum_{ i \in \overline{ \IntSet{\NewLb}{\NewUb} } }  \binom{k}{i} \frac{k - 2i}{k}.
    \end{equation*}

    \noindent Note that for each $i \in \IntSet{0}{k}$,
    $
        \binom{k}{i} \frac{k - 2i}{k} + \binom{k}{k - i} \frac{k - 2(k - i)}{k}
        = 0.
    $
    For each~$i \in \overline{ \IntSet{\NewLb}{\NewUb} }$, if
    ~$i \notin \IntSet{k - \NewUb}{k - \NewLb}$, then it holds that~$k - i \in \overline{ \IntSet{\NewLb}{\NewUb} }$.
    Therefore, 
    $
        \sum_{i \in \overline{ \IntSet{\NewLb}{\NewUb} } \setminus \IntSet{k - \NewUb}{k - \NewLb}}  \binom{k}{i} \frac{k - 2i}{k} = 0
    $, and  
    \begin{align*}
        \cGap
        &=\sum_{ i = \NewLb }^{  \NewUb } \binom{ k }{ i } g(i)  \frac{k - 2i}{k}  
            + \PrOut^* \sum_{ i \in [k - \NewUb, k - \NewLb] } \binom{ k  }{ i  }  \frac{k - 2i}{k} \\
        &= \sum_{ i = \NewLb }^{  \NewUb } \binom{ k }{ i } \PAREN{ g(i) - \PrOut^* } \frac{k - 2i}{k}.
    \end{align*}
    where the second equality follows from that for each~$i \in \IntSet{\NewLb}{\NewUb}$, $k - i \in \IntSet{k - \NewUb}{k - \NewLb}$, and that~$\binom{k}{i} = \binom{k}{k - i}$. 
    Since~$\NewLb = kp - 2 \sqrt{k} < kp \le \NewUb$, 
    it holds that~$\IntSet{\NewUb - 2 \sqrt{k}}{\NewUb - \sqrt{k} / 2} \subset \IntSet{\NewLb}{\NewUb}$.
    Combing with that~$\PrOut^* \le 2^{-k}$ (Inequality~(\ref{ineq: outter probability of modified composed randomizer})), we see
    \begin{equation}
        \label{ineq: sketch proof last 3}
        \cGap \ge \sum_{ i = \NewUb - 2 \sqrt{k} }^{  \NewUb - \sqrt{k} / 2 } \binom{ k }{ i } \PAREN{ g(i) - \frac{1}{2^k} } \frac{k - 2i}{k}.
    \end{equation}
    It is left to prove that the right hand side is lower bounded by~$\Omega(\tildeeps) = \Omega(\eps / \sqrt{k})$. 
    We will show that for each $i \in \IntSet{\NewUb - 2 \sqrt{k}}{\NewUb - \sqrt{k} / 2}$,
    \begin{equation}
        \label{ineq: sketch proof last 2}
        \PAREN{ g(i) - \frac{1}{2^k} } \frac{k - 2i}{k} \in \Omega \PAREN{ \frac{\tildeeps}{2^k} }.
    \end{equation}
    Further, 
    \begin{equation}
        \label{ineq: sketch proof last 1}
        \sum_{ i = \NewUb - 2 \sqrt{k} }^{  \NewUb - \sqrt{k} / 2 } \binom{k}{i} \in \Omega \PAREN{ 2^k }. 
    \end{equation}
    Combining Inequalities~(\ref{ineq: sketch proof last 3}),~(\ref{ineq: sketch proof last 2}), and~(\ref{ineq: sketch proof last 1}) proves~$\cGap \in \Omega(\tildeeps)$. 
    
    \noindent {\it Proving Inequalities~(\ref{ineq: sketch proof last 2}).}
    Observe that
    $
        g(i) 
            \ge g(\NewUb - \sqrt{k} / 2) 
            = e^{ \tildeeps \cdot \sqrt{k} / 2 } \cdot 2^{-k}.
    $ 
    Since~$\NewUb \le k / 2$, 
    $
        \frac{k - 2i}{k}
        \ge
        \frac{k - 2(\NewUb - \sqrt{k}/2)}{k}
        \ge 
        \frac{k - 2(k / 2 - \sqrt{k}/2)}{k}
    $
    $
        =
        \frac{1}{\sqrt{k}}.
    $
    Combined, we obtain
    \begin{equation}
        \left( g(i) - \frac{1}{2^k} \right) \frac{k - 2i}{k} 
        \ge 
        \left( e^{ \tildeeps \cdot \sqrt{k} / 2 }  - 1 \right) \cdot \PAREN{ \frac{1}{2} }^k \cdot \frac{1}{\sqrt{k}}   
        \ge 
        \frac{ \tildeeps }{2} \cdot \PAREN{ \frac{1}{2} }^k.       
    \end{equation}

    \noindent {\it Proving Inequalities~(\ref{ineq: sketch proof last 1}).}
    At a high level, the claim follows from that the summation consists of~$\Omega \paren{ \sqrt{k} }$ terms, each of size~$\Omega \paren{ 2^k / \sqrt{k} }$,
    details in the Appendix.
    \hfill $\square$

\section{Related Works} \label{sec: review}

\noindent {\bf Central Model}. 
Data analysis under {\it continual observation} has been studied in the central model of differential privacy~\citep{DworkNPR10, ChanSS11}. Here, a trusted curator, to which the clients report their true data,  perturbs and releases the aggregated data. 
Given a stream of Boolean ($0$-$1$) values, differentially private frequency estimation algorithms were independently proposed by both~\citeauthor{DworkNPR10}~\citep{DworkNPR10} and~\citeauthor{ChanSS11}~\citep{ChanSS11}: at each time period, they an estimate of the number of~$1$s appearing so far. 
These algorithms guarantee an error of~$O \paren{ 1 / \eps \cdot \log^{1.5} t }$ (omitting failure probability) at each time~$t$.

\vspace{2mm}
\noindent {\bf Local Model}.
To avoid naively repeating algorithms designed for one-time computation, the \emph{memoization} technique was proposed and deployed for continual collection of counter data~\citep{EPK14, DKY17}, where noisy answers for all elements in the domain are memoized. 
However, as pointed out by~\citep{DKY17}, this technique can violate differential privacy. 
Recent works~\citep{erlingsson2020amplification, JRUW18, zhou2021locally} also propose to exploit the potential sparsity of users' data, if it changes infrequently.
They differ in the frequency 
when the algorithm needs to give a prediction.

\noindent {\it Online Setting}. 
Our algorithmic framework is inspired by the work of~\citeauthor{EFMRTT19}~\cite{erlingsson2020amplification} for the online setting, where the server is required to report the estimate at each time step.
Here we provide a succinct description of their protocol,  in the notation and framework of our paper.
The protocol by~\citeauthor{EFMRTT19} requires
an additional sampling step: each user~$u$ samples uniformly and keeps only one non-zero coordinate of~$\UsrDerivative{u}$, and sets all other non-zero coordinates to~$0$.
Therefore, there can be at most one non-zero partial sum with sampled order~$h_u$.
This partial sum is perturbed by the basic randomizer~$\cR$ (Equation~(\ref{equa: def of basic randomizer})) with~$\tildeeps = \eps / 2$, resulting in a~$\paren{ \cGap }^{-1} \in \Omega \paren{ 1 / \eps }$.
However, due to the additional sampling step, the server side estimator of~$\EstPartialsum{ \dyadicInterval{h}{j} }$ (Algorithm~\ref{algo: server}, line~\ref{line: algo: server: update partial sum}) needs to be multiplied by an additional factor of~$k$.
Via similar analysis to Lemma~\ref{lemma: error of single estimator for basic private algorithm},
the error guarantee of their protocol involves a factor of~$k$, instead of~$\sqrt{k}$.

\noindent {\it Offline Setting}. 
The recent independent parallel work by~\citeauthor{zhou2021locally}~\citep{zhou2021locally}
considers the problem in the offline setting, where the server is required to report only after it has collected all the data.
They describe a protocol that provides error guarantee of~$O ( (1 / \eps) \cdot \sqrt{  k \cdot (\log n / \beta) \cdot n  \cdot \log ( d /  \beta ) } )$.
Their client-side algorithm involves hashing the coordinates of user data into a table, and reporting a perturbed version of the hash table to the server. 
Since the value of the entries of the hash table depends on all coordinates of the user data, it is unclear how to convert this algorithm into an online one. 

\noindent {\it Batch Reporting}. The work by~\citeauthor{JRUW18}~\citep{JRUW18} lies between the online and offline settings: time steps are batched into \emph{epochs}, and the server is required to report at each epoch. 
Further, it assumes that users' data are sampled from some (unknown) distributions, 
and analyzes its protocol's performance based on this assumption.
Due to this difference in the problem setting, the performance guarantees of the protocol in~\citep{JRUW18}
cannot be compared directly to those provided in~\cite{erlingsson2020amplification, zhou2021locally} and our paper.
However, the performance in~\citep{JRUW18} also degrades only linearly, instead of sub-linearly, with the number of changes in the underlying data distributions.

\vspace{2mm}
\noindent {\bf Composed Randomizer}.
The composed randomizer,~$\tilde{\cR}$, presented in this paper builds on the composed randomizer proposed by~\citeauthor{BNS19}~\citep{BNS19}.
Their design
focused on preserving the statistical distance between the distribution of the output of the composed randomizer, and joint distribution of~$k$ independent randomized responses.
This difference in the problem setting prompts the non-trivial changes in parameters, assumptions and analysis.
The \citeauthor{BNS19} proof~\citep{BNS19} relies extensively on concentration and anti-concentration inequalities.
Their design can only achieve
$        
    \cGap \in 
    O \Big( {\eps} / {\sqrt{k \ln ( k / \eps) } } + \big( \frac{ \eps }{ k \ln ( k / \eps) } \big)^{2 / 3} \Big)
$ 
(detailed discussion and proof are in Appendix~\ref{appendix: subsec: bun randomizer}). 
When the first term dominates~(e.g, when $k \ge 1 / \eps^2$), this simplifies to~$\cGap \in O ( {\eps} / {\sqrt{k \ln ( k / \eps) } } )$, which implies that~$(\cGap)^{-1} \in \Omega ( {\sqrt{k \ln ( k / \eps) } } / {\eps} )$.
If we apply this composed randomizer to our framework, according to Lemma~\ref{lemma: error of single estimator for basic private algorithm}, this leads to an error that scales at least with~${\sqrt{k \ln ( k / \eps) } } / \eps$.
In comparison, our composed randomizer reduces this to at most~$\sqrt{k} / \eps$. 
Further, \citeauthor{BNS19}'s original design~\citep{BNS19} applies only to offline inputs. 
We overcome this limitation by including our pre-computation technique so the algorithm becomes online.
\begin{acks}
    We thank the anonymous reviewers for their detailed feedback that helped us improve our work.
    Olga Ohrimenko is in part supported by a Facebook research award.
    Hao Wu is supported by an Australian Government Research Training Program (RTP) Scholarship.
\end{acks}

\bibliographystyle{ACM-Reference-Format}
\bibliography{reference}

\appendix
\section{Appendix}

\begin{fact}[Hoeffding's Inequality \citep{DL01}] \label{fact: hoeffding} Let $Y_1, \ldots$, $Y_n$ be independent real-valued random variables such that that $|Y_i| \in [a_i, b_i]$, $\forall i \in [n]$ with probability one. Let $Y = \sum_{i \in [n]} Y_i$, then for every~$\eta \ge 0$:
  \begin{align*}
    \P{ Y - \E{ Y } \ge \eta } & \le \exp \left( - \frac{2 \eta^2 }{ \sum_{i \in [n] } (b_i - a_i)^2 } \right)\,, \text{ and} \\
    \P{ \E{ Y } - Y \ge \eta } & \le \exp \left( - \frac{2 \eta^2 }{ \sum_{i \in [n] } (b_i - a_i)^2 } \right)\,.
  \end{align*}
\end{fact}

\begin{corollary} \label{corollary: confident interval of hoeffding}
  Let $Y_1, \ldots, Y_n$ be independent real-valued random variables such that that $|Y_i| \in [-1, 1], \forall i \in [n]$ with probability one.
  Let $Y = \sum_{i \in [n]} Y_i$. Then, for $\beta \ge 0$, with probability at most $\beta$, it holds that
  $$
    | Y - \E{ Y } | \ge \sqrt{ 2 n \cdot \ln (2 / \beta) }.
  $$
\end{corollary}

\begin{fact}[Stirling's Approximation \cite{R55, MN09}] \label{fact: stirling} For $n = 1, 2, ...$
    \begin{equation} \label{ineq: stirling}
            \sqrt{2 \pi n} \left( \frac{n}{e} \right)^n \exp \left( \frac{1}{12n + 1} \right) 
            \le n! 
            \le 
            \sqrt{2 \pi n} \left( \frac{n}{e} \right)^n \exp \left(  \frac{1}{12n} \right).
    \end{equation}
\end{fact}

\begin{fact}[Entropy Bound~\cite{topsoe2001bounds}]
    Let~$H(x) \doteq - x \log x - (1-x) \log (1 - x)$, $\forall x \in [0, 1]$ be the binary entropy function. Then 
    \begin{equation*}
        4x(1 - x) \le H(x).
    \end{equation*}
\end{fact}

\begin{corollary} \label{corollary: entropy inequality} \label{corollary: binary entropy inequality}
  For each $x \in [-1/2, 1/ 2]$, 
    \begin{equation} 
        1 - 4 x^2 \le H( 1 / 2 - x). 
    \end{equation}
\end{corollary}

\subsection{Proofs for Section~\ref{sec: our randomizer}}

\subsubsection{\bf Lemma~\ref{lemma: probability of modified composed randomizer}}
\label{appendix: subsubsec lemma: probability of modified composed randomizer}

\begin{proof}[Proof of Lemma~\ref{lemma: probability of modified composed randomizer}.]
    Recall that $p = 1 / ( e^{ \tildeeps } + 1)$, $1 - p = e^{\tildeeps} p$, 
    \begin{equation}
        \label{appendix eq: def of g}
        g(i) \doteq p^i (1 - p)^{k - i}
        = p^k \cdot e^{ \tildeeps \cdot (k - i) }, \forall i \in \IntSet{0}{k}, 
    \end{equation}
    and
    \begin{align*}
        \pavg \doteq p^{kp} (1 - p)^{k - kp}
        = g(kp)
        = p^k \cdot e^{ \tildeeps \cdot (k - kp) }.
    \end{align*}
    By definitions of $g$ and $\pavg$, for each $j \in \mbbZ$, it holds that
    \begin{equation} \label{eq: prob exp via pavg}
        \begin{array}{rll}
            g(kp + j) & = \pavg \cdot e^{- \tildeeps \cdot j},\, & \text{if}\, 0 \le kp + j \le k,\, \\
            g(kp - j) & = \pavg \cdot e^{\tildeeps \cdot j},\,   & \text{if}\, 0 \le kp - j \le k.\,
        \end{array}
    \end{equation}
    To prove the lemma, we consider $s  \in \Annulus{b}$ and $s  \in \set{-1, 1}^k \setminus \Annulus{b}$ separately.
    We show that for each~$b \in \set{-1, 1}^k$, it holds that
    \begin{align}
        \label{appendix ineq: inner probability of modified composed randomizer}
        \P{ \tilde{\cR} (b) = s  } 
        &\in \left[
            1 / 2^k , e^{ 2 \tildeeps \sqrt{ k } } \cdot
            \pavg
            \right],\, \quad \forall s \in \Annulus{b},
        \\
        \label{appendix ineq: outter probability of modified composed randomizer}
        \P{ \tilde{\cR} ( b ) = s  } 
        &\in \left[
            e^{-3 \tildeeps \sqrt{ k } } \cdot \pavg,
            1 / 2^k
            \right],\, \quad \forall s \notin \Annulus{b}.
    \end{align}
    
    \vspace{1mm}
    \noindent {\bf Proof of Inequality~(\ref{appendix ineq: inner probability of modified composed randomizer}).}
    Since~$\NewLb = kp - 2 \sqrt{k}$, 
    $
        g(\NewLb) = e^{ 2 \tildeeps \sqrt{ k } } \cdot \pavg.
    $
    We will also prove that
    \begin{equation}
        g(kp) \ge 2^{-k} \ge g( k / 2). 
    \end{equation}
    and that~$g(\NewUb) = 2^{-k}$.
    As~$g$ is a decreasing function, it follows that~$\NewUb \in [kp, k / 2]$.

    Recall that for an input~$b \in \set{-1, 1}^k$,~$\tilde{\cR}$ first generates~$b' = \cR(b) \doteq \PAREN{\cR(b_1), \ldots, \cR(b_k) }$. 
    If~$b' \in \Annulus{b}$,~$\tilde{\cR}$ outputs it directly. 
    Therefore, for each~$s \in \Annulus{b}$, it holds that
    $$
        \P{\tilde{\cR} (b) = s  } = \P{ \cR(b) = s  } = g( \norm{ b - s }_0 ).
    $$
    Noting that~$\norm{ b - s }_0 \in \IntSet{\NewLb}{\NewUb}$, we obtain
    \begin{equation*}
        2^{-k} = g(\NewUb) \le \P{\tilde{\cR} (b) = s  } = g( \norm{ b - s }_0 ) \le  g( \NewLb ) = e^{ 2 \tildeeps \sqrt{ k } } \cdot \pavg, 
    \end{equation*}
    which proves Inequality~(\ref{appendix ineq: inner probability of modified composed randomizer}).
    
    \noindent {\it Proving~$g(kp) \ge 1 / 2^k \ge g(k / 2)$.}
    Via convexity of the function $y = -\log x$, we have
    \begin{align*}
        \frac{ \log ( g(kp) ) }{ k }
         & =  - p \log \frac{1}{p} - (1 - p) \log \frac{1}{1 - p}                   \\
         & \ge - \log \left( p \cdot \frac{1}{p} + (1 - p) \cdot \frac{1}{1 - p} \right)
        = -1.
    \end{align*}
    Via concavity of the function $y = \log x$, we have \begin{align*}
        \frac{ \log g(k / 2) }{ k }
         & = \frac{1}{2} \cdot \log p +  \frac{1}{2} \cdot \log (1 - p)              \\
         & \le \log \left( \frac{1}{2} \cdot p + \frac{1}{2} \cdot (1 - p)  \right) = -1.
    \end{align*}
    It follows that
    \begin{equation} \label{ineq: pavg at least 1 over 2k}
        \pavg = g(kp) \ge 1 / 2^k \ge g(k / 2).
    \end{equation}
    
    \noindent {\it Proving~$g(\NewUb) = 1 / 2^k$.}
    Via Equation~(\ref{appendix eq: def of g}) and the definitions of~$p$ and~$\NewUb$, 
    \begin{align*}
        g(\NewUb)
         & = p^k \cdot \exp \PAREN{ \tildeeps \cdot \PAREN{ k - \frac{k}{ \tildeeps } \cdot \ln \frac{2 e^{\tildeeps} }{ e^{\tildeeps} + 1 } } }                                          \\
         & = \PAREN{ \frac{1 }{ e^{ \tildeeps } + 1 } }^k \cdot \exp \PAREN{ \tildeeps \cdot \PAREN{ k - \frac{k}{ \tildeeps } \cdot \ln \frac{2 e^{\tildeeps} }{ e^{\tildeeps} + 1 } } } \\
         & = \exp \PAREN{ k \ln \PAREN{ \frac{1 }{ e^{ \tildeeps } + 1 } } + k \tildeeps - k \cdot \ln \frac{2 e^{\tildeeps} }{ e^{\tildeeps} + 1 } }                                     \\
         & = \exp \PAREN{ k \ln \PAREN{ \frac{ e^{ \tildeeps } }{ e^{ \tildeeps } + 1 } } - k \cdot \ln \frac{2 e^{\tildeeps} }{ e^{\tildeeps} + 1 } }                                    
         = 1 / 2^k.
    \end{align*}
    Therefore, $\NewUb \in [kp, k / 2]$.

    \vspace{1mm}
    \noindent {\bf Proof of Inequality~(\ref{appendix ineq: outter probability of modified composed randomizer}).}
    Let $\overline{ \IntSet{\NewLb}{\NewUb} } \doteq \IntSet{0}{k} \setminus \IntSet{\NewLb}{\NewUb}$.
    Denote~$\cR(b) \doteq \PAREN{\cR(b_1), \ldots, \cR(b_k) }$. 
    First, observe
    $$
        \P{ \cR(b) \notin \Annulus{b} } = \sum_{ i \in \overline{ \IntSet{\NewLb}{\NewUb} } } \binom{k}{i} g(i).
    $$
    As each $s \notin \Annulus{b}$ is outputted uniformly at random by~$\tilde{\cR}$ when~$b' = \cR (b) \notin \Annulus{b}$, it holds that 
    \begin{equation}
        \P{ \tilde{\cR} (b ) = s } = \PrOut^* \doteq
        \frac{
            \sum_{ i \in \overline{ \IntSet{\NewLb}{\NewUb} } } \binom{k}{i} g(i)
        }
        {
            \sum_{ i \in \overline{ \IntSet{\NewLb}{\NewUb} } } \binom{k}{i}
        }.
    \end{equation}
    \noindent We upper bound and lower bound $\PrOut^*$ separately.

    \noindent \textbf{Upper Bound For $\PrOut^*$.}
    As~$g$ is decreasing and~$g(\NewUb) = 1 / 2^k$, it holds that
    $$
        \sum_{i \in \IntSet{\NewLb}{\NewUb} } \binom{k}{i} g(i) \ge \sum_{i \in \IntSet{\NewLb}{\NewUb} } \binom{k}{i} \cdot 1 / 2^k.
    $$
    Combing with the fact that
    $$
        \sum_{i \in \IntSet{0}{k}  } \binom{k}{i} g(i)
        = 1
        = \sum_{i \in \IntSet{0}{k}  } \binom{k}{i} \cdot 1 / 2^k,
    $$
    we obtain
    \begin{align*}
        \sum_{i \in \overline{ \IntSet{\NewLb}{\NewUb} }  } \binom{k}{i} g(i) \le \sum_{i \in \overline{ \IntSet{\NewLb}{\NewUb} }  } \binom{k}{i} \cdot 1 / 2^k.
    \end{align*}
    It concludes that $\PrOut^* \le 1 / 2^k$.

    \vspace{1mm}
    \noindent \textbf{Lower Bound For $\PrOut^*$.}
    Since $\NewLb + \NewUb \le kp - 2 \sqrt{k} + k / 2 < k$, it holds that $k - \NewLb > \NewUb$.
    Hence, we can partition the numerator and denominator of $\PrOut^*$ into three parts:

    \begin{equation*}
        \numerator{1} \doteq \sum_{ i = 0 }^{ \NewLb - 1 } \binom{k}{i} g(i),
        \numerator{2} \doteq \sum_{ i = \NewUb + 1 }^{ k - \NewLb} \binom{k}{i} g(i),
        \numerator{3} \doteq \sum_{ i = k - \NewLb + 1 }^{ k} \binom{k}{i} g(i).
    \end{equation*}
    \begin{align*}
        \denominator{1} \doteq \sum_{ i = 0 }^{ \NewLb - 1 } \binom{k}{i},\,
        \denominator{2} \doteq \sum_{ i = \NewUb + 1 }^{ k - \NewLb} \binom{k}{i},\,
        \denominator{3} \doteq \sum_{ i = k - \NewLb + 1 }^{ k} \binom{k}{i}.
    \end{align*}
    We will prove that
    \begin{equation}
        \label{ineq: nu2 vs de 2}
        \frac{\numerator{2}}{\denominator{2}}
        \ge \pavg \cdot e^{- \tildeeps \cdot \PAREN{ k(1 - 2p) + 2 \sqrt{k} } },
    \end{equation}
    and
    \begin{equation}
        \label{ineq: nu1 nu3 vs de 1 de3}
        \frac{\numerator{1} + \numerator{3}}{\denominator{1} + \denominator{3}}
        \ge \pavg \cdot e^{- \tildeeps \cdot k (1 - 2p)  }.
    \end{equation}
    It follows that
    \begin{align*}
        \PrOut^*
        = \frac{ \numerator{1} + \numerator{2} +  \numerator{3} }{ \denominator{1} + \denominator{2} + \denominator{3} }
        \ge \pavg \cdot e^{- \tildeeps \cdot \PAREN{ k(1 - 2p) + 2 \sqrt{k} } }.
    \end{align*}
    Since $\PAREN{e^x - 1} / \PAREN{e^x + 1} \le x / 2$ for all $x \ge 0$, we have
    $$
        k (1 - 2p) = k \cdot ( e^{\tildeeps} - 1 ) / ( e^{\tildeeps} + 1  ) \le k \cdot \tildeeps / 2.
    $$
    Via the assumption that~$\eps \le 1$, we get~$\tildeeps = \eps / \paren{ 5 \sqrt{k} } \le 1 / \sqrt{k}$, and~$k (1 - 2p) \le \sqrt{k}$.
    Therefore,
    $$
        \PrOut^*
        \ge
        e^{- \tildeeps \cdot \PAREN{ 3 \sqrt{k} } }  \cdot \pavg.
    $$

    \vspace{1mm}
    \noindent {\it Proof For Inequality~(\ref{ineq: nu2 vs de 2}).}
    Via Equation~(\ref{eq: prob exp via pavg}),
    \begin{align*}
        g(k - \NewLb)
         & = g(k - k p + 2 \sqrt{k}) = g( k p + (k - 2 k p) + 2 \sqrt{k})         \\
         & = \pavg \cdot e^{- \tildeeps \cdot \PAREN{ k(1 - 2p) + 2 \sqrt{k} } }.
    \end{align*}
    On the other hand, as $g$ is a decreasing function,
    \begin{align*} \label{ineq: lower bound num2}
        \numerator{2}
        \ge \sum_{ i = \NewUb + 1 }^{ k - \NewLb} \binom{k}{i} g(k - \NewLb)
        \ge \denominator{2} \cdot \pavg \cdot e^{- \tildeeps \cdot \PAREN{ k(1 - 2p) + 2 \sqrt{k} } }.
    \end{align*}

    \vspace{1mm}
    \noindent {\it Proof For Inequality~(\ref{ineq: nu1 nu3 vs de 1 de3}).}
    Both $\numerator{1}$ and $\numerator{3}$ contain $\NewLb$ summands.
    We pair up these terms in $\numerator{1} + \numerator{3}$ as
    \begin{align*}
          & \sum_{ j = 1 }^{ \NewLb } \left( \binom{k}{\NewLb - j} g( \NewLb - j) + \binom{k}{k - \NewLb + j } g(k - \NewLb + j ) \right) \\
        = & \sum_{ j = 1 }^{ \NewLb } \binom{k}{\NewLb - j} \cdot \PAREN{ g( \NewLb - j ) +  g( k - \NewLb + j ) }.
    \end{align*}

    \noindent Via Equation~(\ref{eq: prob exp via pavg}) and the definition of $\NewLb$,
    $$
        g(\NewLb - j)
        = g(kp - 2 \sqrt{k} - j) = \pavg \cdot e^{ \tildeeps \cdot \PAREN{ 2 \sqrt{k} + j } }.
    $$
    Similarly, using that $k - \NewLb = kp + k(1 - 2p) + 2 \sqrt{k}$, we get
    \begin{align*}
        g( k - (\NewLb - j) )
         & = g( kp + k(1 - 2p) + 2 \sqrt{k} + j )                               \\
         & = \pavg \cdot e^{- \tildeeps \cdot \PAREN{ k(1 - 2p) + 2 \sqrt{k} + j } }.
    \end{align*}
    Via the inequality that $x + 1 / x \ge 2, \forall x > 0$,
    \begin{align*}
         & g(\NewLb - j) + g( k - (\NewLb - j) )                                                                                                                  \\
         & = \pavg \cdot e^{ \tildeeps \cdot ( 2 \sqrt{k} + j ) } +  \pavg \cdot e^{- \tildeeps \cdot ( k(1 - 2p) + 2 \sqrt{k} + j ) }                            \\
         &\ge \pavg \cdot \PAREN{ e^{ \tildeeps \cdot ( 2 \sqrt{k} + j ) } +  e^{- \tildeeps \cdot ( 2 \sqrt{k} + j ) } } \cdot e^{- \tildeeps \cdot k (1 - 2p)  } \\
         &\ge \pavg \cdot 2 \cdot e^{- \tildeeps \cdot k (1 - 2p)  }.
    \end{align*}
    It concludes that
    \begin{align}
        \numerator{1} + \numerator{3}
         & \ge \sum_{ j = 1 }^{ \NewLb} \binom{k}{\NewLb - j} \cdot \pavg \cdot 2 \cdot e^{- \tildeeps \cdot k (1 - 2p)  }.
    \end{align}
    Noting that $\denominator{1} + \denominator{3} = 2 \cdot \sum_{ j = 1 }^{ \NewLb} \binom{k}{\NewLb - j}$ finishes the proof.

\end{proof}

\subsubsection{\bf Lemma~\ref{lemma: lower bound of csvr}}
\label{appendix: subsubsec lemma: lower bound of csvr}

\begin{proof}[Proof of Lemma~\ref{lemma: lower bound of csvr}]
    By symmetry of~$\tilde{\cR}$, it suffices to prove that Inequality~(\ref{lemma: lower bound of csvr}) holds for the first coordinate of~$\tilde{b}$.
    
    We study the probability of $\P{ \tilde{b}_1 = b_1 }$.
    Recall that for an input~$b \in \set{-1, 1}^k$,~$\tilde{\cR}$ first generates~$b' = \cR(b) \doteq \paren{\cR(b_1), \ldots, \cR(b_k) }$. 
    If~$b' \in \Annulus{b}$,~$\tilde{\cR}$ outputs it directly. 
    Therefore, for each~$s \in \Annulus{b}$, it holds that
    $$
        \P{\tilde{b} = s  } = \P{ \cR(b) = s  } = g( \norm{ b - s }_0 ).
    $$
    For each~$i \in \IntSet{\NewLb}{\NewUb}$, consider the set~$\{ s \in \set{-1, 1}^k : \norm{s - b}_0 = i \} \subset \Annulus{b}$.
    The set contains~$\binom{k}{i}$ sequences such that each sequence~$s$ in the set is outputted by~$\tilde{\cR}$ with probability
    $$  
        \P{\tilde{b} = s  } = \P{ \cR(b) = s  } = g( \norm{ b - s }_0 ) = g(i).
    $$
    
    Note that each sequence in the set~$\{ s \in \set{-1, 1}^k : \norm{s - b}_0 = i \} \subset \Annulus{b}$ differs in~$i$ coordinates from~$b$.
    If we pick a sequence uniformly at random from this set, with probability~$i / k$, we obtain a sequence~$s$ such that $s_1 \neq b_1$.
    It follows that the fraction of sequences in the set whose first coordinate equals~$b_1$ is given by~$\paren{k - i} / k$.
    Therefore,
    $$
        \Pr \left[ \tilde{b}_1 = b_1 , \, \norm{ \tilde{b} - b } = i \right] = \binom{ k }{ i } g(i) \frac{k - i}{k}.
    $$
    Similarly, we can prove that for each~$i \in \overline{ \IntSet{\NewLb}{\NewUb} }$, 
    $$
        \Pr \left[ \tilde{b}_1 = b_1 , \, \norm{ \tilde{b} - b } = i \right] 
        = \binom{ k }{ i }  \PrOut^* \frac{k - i}{k}.
    $$
    Summing over~$i \in \IntSet{0}{k}$, we get 
    \begin{equation*}
        \P{ \tilde{b}_1 = b_1 }
        = \sum_{ i = \NewLb }^{  \NewUb } \binom{ k }{ i } g(i) \frac{k - i}{k}
        + \PrOut^* 
        \sum_{ i \in \overline{ \IntSet{\NewLb}{\NewUb} } } \binom{ k  }{ i  } \cdot \frac{k - i}{k},\,
    \end{equation*}
    where~$\overline{ \IntSet{\NewLb}{\NewUb} } = \IntSet{0}{\NewLb - 1} \cup \IntSet{\NewUb + 1}{k}$.
    Similarly,  
    \begin{equation*}
        \P{ \tilde{b}_1 = -b_1 } = \sum_{i = \NewLb }^{  \NewUb } \binom{ k }{ i } g(i) \frac{ i }{ k }
        + \PrOut^* 
        \sum_{ i \in \overline{ \IntSet{\NewLb}{\NewUb} } } \binom{ k  }{ i  } \cdot \frac{i}{k}.
    \end{equation*}
    It follows that
    \begin{equation}
        \cGap = \sum_{ i = \NewLb }^{ \NewUb } g(i)  \binom{k}{i}  \frac{k - 2i}{k}
        + \sum_{ i \in \overline{ \IntSet{\NewLb}{\NewUb} } } \PrOut^* \binom{k}{i}  \frac{k - 2i}{k}.
    \end{equation}

    \noindent For each $i \in \IntSet{0}{k}$, it holds that
    $$
        \binom{k}{i} \frac{k - 2i}{k} + \binom{k}{k - i}  \frac{k - 2 \PAREN{k - i}}{k}
        = \binom{k}{i} \cdot \PAREN{ 2 -  2 \cdot \frac{i + k - i}{k} }
        = 0.
    $$
    For each~$i \in \overline{ \IntSet{\NewLb}{\NewUb} }$, if
    ~$i \notin \IntSet{k - \NewUb}{k - \NewLb}$, then $k - i \in \overline{ \IntSet{\NewLb}{\NewUb} }$.
    Therefore, by pairing up~$\binom{k}{i} \frac{k - 2i}{k}$ with~$\binom{k}{k - i} \frac{k - 2(k - i)}{k}$ for each~$i \in \overline{ \IntSet{\NewLb}{\NewUb} } \setminus \IntSet{k - \NewUb}{k - \NewLb}$, 
    we obtain
    \begin{align*}
        \sum_{ i \overline{ \IntSet{\NewLb}{\NewUb} } } \binom{ k  }{ i  } \cdot  \frac{k - 2i}{k}
        =
        \sum_{ i \in [k - \NewUb, k - \NewLb] } \binom{ k  }{ i  } \cdot  \frac{k - 2i}{k}.
    \end{align*}

    \noindent Therefore,
    \begin{align*}
        \cGap
         & =\sum_{ i = \NewLb }^{  \NewUb } \binom{ k }{ i } g(i) \frac{k - 2i}{k}        
            + \PrOut^* \sum_{ i \in [k - \NewUb, k - \NewLb] } \binom{ k  }{ i  } \frac{k - 2i}{k}                                \\
         & = \sum_{ i = \NewLb }^{  \NewUb } \binom{ k }{ i } \left( g(i)  \frac{k - 2i}{k} + \PrOut^*  \frac{k - 2 \PAREN{k - i}}{k} \right) \\
         & = \sum_{ i = \NewLb }^{  \NewUb } \binom{ k }{ i } \PAREN{ g(i) - \PrOut^* } \frac{k - 2i}{k}.
    \end{align*}
    Since~$\NewUb > \NewLb = kp - 2 \sqrt{k}$, it holds that~$\IntSet{\NewUb - 2 \sqrt{k}}{\NewUb - \sqrt{k} / 2} \subset \IntSet{\NewLb}{\NewUb}$.
    According to Inequality~(\ref{appendix ineq: outter probability of modified composed randomizer}), we have~$\PrOut^* \le 1 / 2^k$.
    Via that~$g(i)$ is a decreasing function,
    Equality~(\ref{eq: prob exp via pavg}) and Inequality~(\ref{ineq: pavg at least 1 over 2k}), we get that for each $i \in \IntSet{\NewUb - 2 \sqrt{k}}{\NewUb - \sqrt{k} / 2}$
    $$
        g(i) \ge g(\NewUb - \sqrt{k} / 2)
        = e^{ \tildeeps \cdot \sqrt{k} / 2 } \cdot \pavg
        \ge
        e^{ \tildeeps \cdot \sqrt{k} / 2 } \cdot 1 / 2^k.
    $$
    Further, via that~$\NewUb \le k / 2$, for each $i \in \IntSet{\NewUb - 2 \sqrt{k}}{\NewUb - \sqrt{k} / 2}$,
    $$
        \frac{k - 2i}{k}
        \ge
        \PAREN{ 1 -  2 \cdot \frac{ \NewUb - \sqrt{k} / 2 }{k} }
        =
        \PAREN{ \frac{1}{\sqrt{k}} + \PAREN{ 1 - \frac{\NewUb}{k / 2} } }
        \ge
        \frac{1}{\sqrt{k}}.
    $$
    Now we can lower bound~$\cGap$ by
    \begin{align*}
        \cGap
         & \ge \sum_{i = \NewUb - 2 \sqrt{k}}^{\NewUb - \sqrt{k} / 2}
        \binom{k}{i} \left( e^{ \tildeeps \cdot \sqrt{k} / 2 } \cdot \left( \frac{1}{2} \right)^k - \left( \frac{1}{2} \right)^k \right) \cdot \frac{1}{\sqrt{k}} \\
         & = \sum_{i = \NewUb - 2 \sqrt{k}}^{\NewUb - \sqrt{k} / 2}
        \binom{k}{i} \left( e^{ \tildeeps \cdot \sqrt{k} / 2 }  - 1 \right) \cdot \left( \frac{1}{2} \right)^k \cdot \frac{1}{\sqrt{k}}                      \\
         & \ge
        \sum_{i = \NewUb - 2 \sqrt{k}}^{\NewUb - \sqrt{k} / 2}
        \binom{k}{i} \frac{\tildeeps \sqrt{k}}{2}
        \left( \frac{1}{2} \right)^k
        \frac{1}{ \sqrt{k} }                
        = \frac{\tildeeps }{2} 
        \left( \sum_{i = \NewUb - 2 \sqrt{k}}^{\NewUb - \sqrt{k} / 2}
            \binom{k}{i} \left( \frac{1}{2} \right)^k \right).
    \end{align*}
    To prove that~$\cGap \in \Omega \paren{\tildeeps} = \Omega \paren{ \eps / \sqrt{k} }$, it suffices to prove that
    $$
        \sum_{i = \NewUb - \sqrt{k}}^{\NewUb - \sqrt{k} / 2}
        \binom{k}{i} \left( \frac{1}{2} \right)^k
        \in \Omega (1).
    $$
    Recall that $kp \le \NewUb \le k / 2$.
    Further, we have
    $
        k / 2 - k p 
        = k \cdot \frac{e^{\tildeeps} - 1}{2 \cdot (e^{\tildeeps} + 1) } 
        \le \frac{k}{4} \tildeeps 
        \le \frac{ \sqrt{k} }{4},
    $
    where the first inequality follows from that $(e^x - 1) / (e^x + 1) \le x / 2$ for all $x \ge 0$, the second one from that the assumption that~$\tildeeps \le 1 / \sqrt{k}$.
    Hence, $\NewUb \ge k / 2 - \sqrt{k} / 4$ and
    $$
        \sum_{i = \NewUb - 2 \sqrt{k}}^{\NewUb - \sqrt{k} / 2}
        \binom{k}{i} \left( \frac{1}{2} \right)^k
        \ge
        \sum_{i = k / 2 - 2 \cdot \sqrt{k}}^{ k / 2 - \sqrt{k}  }
        \binom{k}{i} \left( \frac{1}{2} \right)^k.
    $$
    Finally,
    \begin{align*}
        \sum_{i = k / 2 - 2 \cdot \sqrt{k}}^{ k / 2 - \sqrt{k} } \binom{k}{i}
         & = \sum_{i = 0}^{\sqrt{k} } \binom{k}{ k / 2 - \sqrt{k} - i}
        = \binom{k}{ k / 2 - \sqrt{k} }                                             \\
         &\, + \sum_{i = 1}^{\sqrt{k} } \binom{k}{ k / 2 - \sqrt{k} } \prod_{ j \in [i] } \frac{k / 2 - \sqrt{k} - i + j}{ k + 1 - \PAREN{ k / 2 - \sqrt{k} - i + j } }                                                \\
         & \ge \sum_{i = 0}^{\sqrt{k} } \binom{k}{ k / 2 - \sqrt{k} } \left( \frac{k / 2 - \sqrt{k} - i + 1}{ k + 1 - \PAREN{ k / 2 - \sqrt{k} - i  + 1} } \right)^i                                                            \\
         & \ge \sum_{i = 0}^{\sqrt{k} } \binom{k}{ k / 2 - \sqrt{k} } \left( \frac{k / 2 - \sqrt{k} - \sqrt{k} + 1}{ k + 1 - \PAREN{ k / 2 - \sqrt{k} - \sqrt{k}  + 1} } \right)^i                                              \\
         & = \sum_{i = 0}^{\sqrt{k} } \binom{k}{ k / 2 - \sqrt{k} } \left(  1 - \frac{ 4 \sqrt{k} - 1 }{k / 2 + 2 \cdot \sqrt{k} }  \right)^i                                                                                   \\
         & = \binom{k}{ k / 2 - \sqrt{k} } \cdot \frac{ 1 - \left( 1 - \frac{ 4 \sqrt{k} - 1 }{k / 2 + 2 \cdot \sqrt{k} }  \right)^{ \sqrt{k} + 1 } }{ 1 - \left( 1 - \frac{ 4 \sqrt{k} - 1 }{k / 2 + 2 \cdot \sqrt{k} }  \right) }.
    \end{align*}
    Via that~$1 - x \le e^x, \forall x \in \R$, we get
    \begin{align*}
        \frac{ 1 - \left( 1 - \frac{ 4 \sqrt{k} - 1 }{k / 2 + 2 \cdot \sqrt{k} }  \right)^{ \sqrt{k} + 1 } }{ 1 - \left( 1 - \frac{ 4 \sqrt{k} - 1 }{k / 2 + 2 \cdot \sqrt{k} }  \right) }
         & \ge
        \frac{ 1 - \exp \PAREN{ - \frac{ \PAREN{ 4 \sqrt{k} - 1 } \cdot \PAREN{ \sqrt{k} + 1 } }{k / 2 + 2 \cdot \sqrt{k} } } }{ \frac{ 4 \sqrt{k} - 1 }{k / 2 + 2 \cdot \sqrt{k} } } .
    \end{align*}
    When $k \ge 4 \sqrt{k}$, we see
    $
        \exp \PAREN{ - \frac{ \paren{4 \sqrt{k} - 1} \cdot \paren{ \sqrt{k} + 1 } }{k / 2 + 2 \cdot \sqrt{k} } }
        =
        \exp \PAREN{ - \frac{ 4 k + 3 \sqrt{k} - 1 }{k / 2 + 2 \cdot \sqrt{k} } }
    $
    $
        \le
        \exp \PAREN{ - \frac{ 19 \sqrt{k} - 1 }{ 4 \sqrt{k} } }
        \le
        \exp \PAREN{ - \frac{ 18 }{ 4 } },
    $
    and
    $
        \frac{ 4 \sqrt{k} - 1 }{k / 2 + 2 \cdot \sqrt{k} }
        =
        \frac{ 8 - 2 /  \sqrt{k} }{ \sqrt{k} + 4 }
        \le 8 / \sqrt{k}.
    $
    Hence,
    \begin{align*}
        \sum_{i = k / 2 - 2 \sqrt{k}}^{ k / 2 - \sqrt{k} } \binom{k}{i}
         & = \binom{k}{k / 2 - \sqrt{k} } \frac{ \sqrt{k} }{ 8 } \PAREN{ 1 - e^{ - \frac{18}{4} } } 
         \ge \binom{k}{k / 2 - \sqrt{k} } \frac{ \sqrt{k} }{ 9 }.
    \end{align*}
    Using the Stirling's approximation (Fact~\ref{fact: stirling}), we get
    \begin{align*}
        \binom{k}{k / 2 - \sqrt{k} }
         & \ge e^{- 1 / 6} \cdot \sqrt{ \frac{k}{2 \cdot \pi \cdot (k / 2 - \sqrt{k}) \cdot (k / 2 + \sqrt{k}) } }                         \\
         & \quad \cdot \frac{ k^k }{ (k / 2 - \sqrt{k})^{  (k / 2 - \sqrt{k}) } \cdot (k / 2 + \sqrt{k})^{ (k / 2 + \sqrt{ k}) }  }        \\
         & \ge e^{- 1 / 6} \cdot \frac{1}{ \sqrt{k - 4} } \cdot \sqrt{ \frac{ 2 }{ \pi } } \cdot 2^{ k \cdot H ( 1 / 2 - 1 / \sqrt{k} ) },
    \end{align*}
    where~$H(\cdot)$ is the entropy function.
    We obtain that
    \begin{align*}
        \sum_{i = k / 2 - 2 \sqrt{k}}^{ k / 2 - \sqrt{k} } \binom{k}{i} \left( \frac{1}{2} \right)^k
         & \ge \frac{1}{9} \sqrt{\frac{k}{k - 4}} \cdot e^{- 1 / 6} \cdot \sqrt{ \frac{ 2 }{ \pi } } \cdot 2^{ k \cdot \PAREN{ H ( 1 / 2 - 1 / \sqrt{k} ) - 1} }
    \end{align*}
    Via the inequality (Corollary~\ref{corollary: binary entropy inequality}) that $H(1 / 2 - x) \ge 1 - 4x^2, \forall x \in [- 1/ 2, 1 / 2]$, 
    \begin{align*}
        \sum_{i = k / 2 - 2 \sqrt{k}}^{ k / 2 - \sqrt{k} } \binom{k}{i} \left( \frac{1}{2} \right)^k
         & \ge \frac{1}{9} \sqrt{\frac{k}{k - 4}} \cdot e^{- 1 / 6} \cdot \sqrt{ \frac{ 2 }{ \pi } } \cdot 2^{ k \cdot \PAREN{ -4 / k} } 
         \in \Omega(1), 
    \end{align*}
    which finishes the proof. 
\end{proof}

\subsection{Composed Randomizer by~\citep{BNS19}}
\label{appendix: subsec: bun randomizer}
The composed randomizer proposed in~\citep{BNS19} shares the same pseudo-code as ours, but with different parameter setting and assumptions. 
For convenience, we repeat the pseudo-code here. 

\begin{algorithm}[!ht]
    \caption{Composed Randomizer by~\citep{BNS19}}
    \label{algo: composed randomizer by bun}
    \begin{algorithmic}[1]
        \Procedure{Basic Randomizer~$\cR$}{$\zeta$}
        \Require Value $\zeta \in \set{-1, 1 }$.
        \State {\bf return} $-\zeta$ w.p.~$1 / \paren{ e^{ \tildeeps } + 1 }$ and $\zeta$ w.p.~$e^{ \tildeeps } / \paren{ e^{ \tildeeps } + 1 }$.
        \EndProcedure
        \vspace{-4mm}
        \Statex
        \Procedure{Composed Randomizer~$\tilde \cR $}{$b$}
        \Require Vector $b \in \set{-1, 1}^k$.
        \State Sample $b' \leftarrow ( \cR ( b_1 ), \ldots, \cR ( b_k )$.
        \label{algo line: 1 algo:Approximate Composed Algorithm}
        \If{ $b' \notin \Annulus{b}$ }
        \State $b' \uniffrom \set{-1, 1}^k \setminus \Annulus{b}$
        \label{algo line: 3 algo:Approximate Composed Algorithm}
        \EndIf
        \State {\bf return} $b'$
        \EndProcedure
    \end{algorithmic}
\end{algorithm}

Denote $p \doteq 1 / ( e^{ \tildeeps } + 1)$.
The composed randomizer proposed in~\citep{BNS19} sets
\begin{equation}
    \NewLb \doteq k p - \sqrt{ \frac{k}{2} \ln \frac{2}{\lambda}} ,\,\quad
    \NewUb \doteq k p + \sqrt{ \frac{k}{2} \ln \frac{2}{\lambda}}.
\end{equation}
for some additional parameter~$\lambda \in (0, 1)$. According, let 
\begin{equation}
    \Annulus{b} \doteq \set{ s  \in \set{-1, 1}^k : \norm{ b - s }_0 \in     \IntSet{\NewLb}{\NewUb} }
\end{equation}
The pseudo-code of the composed randomizer~\citep{BNS19} is presented in Algorithm~\ref{algo: composed randomizer by bun}.
The work~\citep{BNS19} proves the following fact.

\begin{fact}[Algorithm~\ref{algo: composed randomizer by bun}~\citep{BNS19}]
    Suppose that 
    \begin{equation} \label{ineq: lambda tilde eps and k}
        0 < \lambda < \PAREN{ \tildeeps \sqrt{k} / \PAREN{2 (k + 1)}  }^{2 / 3}\,,
    \end{equation}
    and that
    \begin{equation} \label{eq: eps and tilde eps}
        \eps = 6 \tildeeps \sqrt{k \ln (1 / \lambda) } \le 1,
    \end{equation}
    the algorithm~$\tilde{\cR}$ is~$\eps$-deferentially private, further
    \begin{equation}
        \label{ineq: probability of being inside}
        \P{ \tilde{\cR}(b) \in \Annulus{b} } \ge 1 - \lambda.
    \end{equation}
\end{fact}

We can derived the following constraint between~$\lambda$, $\eps$ and~$k$. 

\begin{theorem}
    If $\lambda$ satisfies Inequality~(\ref{ineq: lambda tilde eps and k}) and Equality~(\ref{eq: eps and tilde eps}), then 
    \begin{align}
        \label{ineq: lower bound of ln 1 over lambda}
        \ln (1 / \lambda) 
                &\in \Omega ( \ln ( k / \eps)   \\
        \label{ineq: lower bound of lambda}
        \lambda &\in O \left( \left(\frac{ \eps }{ k \cdot \ln ( k / \eps) } \right)^{2 / 3} \right)
    \end{align}
\end{theorem}
\begin{proof}
    Substituting $\tildeeps = \frac{\eps}{6 \sqrt{k \ln (1 / \lambda) }}$ into Inequality~(\ref{ineq: lambda tilde eps and k}), we get 
    \begin{align}
        \label{ineq: bound of lambda}
        \lambda^{3 / 2} 
            < \frac{\eps \sqrt{k} }{ 6 \sqrt{k \ln (1 / \lambda) } \cdot 2(k + 1)} 
    \end{align}
    and hence, 
    $
            \frac{3}{2} \ln \lambda + \frac{1}{2} \ln \ln \frac{1}{\lambda} < \ln \frac{ \eps }{ 12 (k + 1) }.
    $
    This proves that $\ln (1 / \lambda) \in \Omega ( \ln ( k / \eps) )$.
    Substituting this back to Inequality~(\ref{ineq: bound of lambda}), we obtain
    $
        \lambda \in O \left( \frac{ \eps }{ k \cdot \ln ( k / \eps) }\right)^{2 / 3}.
    $
\end{proof}

\begin{theorem}
    \label{lemma: lower bound of bun's csvr}
    For a given input~$b$, denote~$\tilde{b} \doteq \tilde{\cR}(b)$ the output of~$\tilde{\cR}$.
    For each input~$b \in \set{-1, 1}^k$, and for each~$i \in \bracket{k}$, it holds that 
    \begin{equation*}
        \P{ \tilde{b}_i = b_i } - \P{ \tilde{b}_i = -b_i } 
        = \cGap 
        \in 
        O \left( \frac{\eps}{\sqrt{k \ln ( k / \eps) } } + \left( \frac{ \eps }{ k \ln ( k / \eps) }\right)^{\frac{2}{3}} \right). 
    \end{equation*}
\end{theorem}

\begin{proof}[Proof of Theorem~\ref{lemma: lower bound of bun's csvr}.]
    Let~$\tilde{b} = \tilde{\cR} (b)$.
    If $\tilde{b}_1 = b_1$, either of the following events happens
    \begin{enumerate}
        \item $\cR(b_1) = b_1$. 
        \item $b' = \cR(b) \notin \Annulus{b}$, and it is replaced with a random sample from $\set{-1, 1}^k \setminus \Annulus{b}$.
    \end{enumerate}
    The former happens with probability $e^{ \tildeeps} / ( e^{ \tildeeps} + 1)$ and the later happens with probability at most $\lambda$ according to Inequality~(\ref{ineq: probability of being inside}). 
    Via union bound, we have 
    $$
        \P{ \tilde{b}_1 = b_1 } 
        \le e^{ \tildeeps} / ( e^{ \tildeeps} + 1) + \lambda. 
    $$
    Since $\P{ \tilde{b}_1 = b_1} +  \P{ \tilde{b}_1 = - b_1 } = 1$, 
    the above inequalities imply
    \begin{align*}
        \cGap 
            =\P{ \tilde{b}_1 = b_1 } - \P{ \tilde{b}_1 = -b_1 } 
            \le \PAREN{ e^{ \tildeeps} - 1 } / \PAREN{ e^{ \tildeeps} + 1 } + 2 \cdot \lambda. 
    \end{align*}
    By Inequality~(\ref{eq: eps and tilde eps}) and~(\ref{ineq: lower bound of ln 1 over lambda}), we have 
    $$
        \tildeeps = \frac{\eps}{6 \sqrt{k \ln (1 / \lambda) }} 
        \in 
        O \left( 
            \frac{\eps}{ \sqrt{k \ln (k / \eps) }} 
        \right).
    $$
    Combing with Inequality~(\ref{ineq: lower bound of lambda}), we have
    \begin{align*}
        \cGap 
        &\le \PAREN{ e^{ \tildeeps} - 1 } / \PAREN{ e^{ \tildeeps} + 1 } + 2 \cdot \lambda \\
        &\in 
        O \left( \frac{\eps}{\sqrt{k \ln ( k / \eps) } } + \left( \frac{ \eps }{ k \ln ( k / \eps) }\right)^{2 / 3} \right). 
    \end{align*}
\end{proof}

\end{document}